\documentclass[conference]{IEEEtran}

\usepackage{algorithmic}
\usepackage{array}
\usepackage[caption=false,font=footnotesize,labelfont=sf,textfont=sf]{subfig}
\usepackage{textcomp}
\usepackage[hidelinks]{hyperref}
\usepackage{url}
\usepackage{verbatim}
\usepackage{graphicx}
\usepackage{amsthm}
\theoremstyle{definition}
\newtheorem{definition}{Definition}
\usepackage{amsmath}
\usepackage{amssymb}
\theoremstyle{plain}
\newtheorem{lemma}{Lemma}
\newcolumntype{L}[1]{>{\centering\arraybackslash}p{#1}}

\newcommand{\cmark}{\textcolor{green!60!black}{\ding{51}}} 			% 绿色对勾
\newcommand{\xmark}{\textcolor{red}{\ding{55}}}           			% 红色叉号
\newcommand{\qmark}{\textcolor{orange}{\textbf{$\otimes$}}} 
\usepackage{booktabs,makecell,xcolor,pifont,threeparttable}
\usepackage{multirow}
\usepackage[ruled,vlined]{algorithm2e}

\begin{document}
		\title{RAIN: Secure and Robust Aggregation under Shuffle Model of Differential Privacy}

		\author{\IEEEauthorblockN{Yuhang Li\IEEEauthorrefmark{1},
				Yajie Wang\IEEEauthorrefmark{1},
				Xiangyun Tang\IEEEauthorrefmark{2},
				Peng Jiang\IEEEauthorrefmark{1},
				Yu-an Tan\IEEEauthorrefmark{1},
				Liehuang Zhu\IEEEauthorrefmark{1}}
			\IEEEauthorblockA{\IEEEauthorrefmark{1}Beijing Institute of Technology, Beijing, China\\
				Email: \{yuhangl, wangyajie19, pengjiang, tan2008, liehuangz\}@bit.edu.cn}
			\IEEEauthorblockA{\IEEEauthorrefmark{2}Minzu University of China, Beijing, China\\Email: xiangyunt@muc.edu.cn}}
		

		% \IEEEoverridecommandlockouts
		% \makeatletter\def\@IEEEpubidpullup{6.5\baselineskip}\makeatother
		% \IEEEpubid{\parbox{\columnwidth}{
				% 		Network and Distributed System Security (NDSS) Symposium 2026\\
				% 		23 - 27 February 2026 , San Diego, CA, USA\\
				% 		ISBN 979-8-9919276-8-0\\  
				% 		https://dx.doi.org/10.14722/ndss.2026.[23$|$24]xxxx\\
				% 		www.ndss-symposium.org
				% }
			% \hspace{\columnsep}\makebox[\columnwidth]{}}

		\maketitle
		
		\begin{abstract}
			Secure aggregation is a foundational building block of privacy-preserving learning, yet achieving robustness under adversarial behavior remains challenging. Modern systems increasingly adopt the shuffle model of differential privacy (Shuffle-DP) to locally perturb client updates and globally anonymize them via shuffling for enhanced privacy protection.
			However, these perturbations and anonymization distort gradient geometry and remove identity linkage, leaving systems vulnerable to adversarial poisoning attacks. 
			Moreover, the shuffler, typically a third party, can be compromised, undermining security against malicious adversaries.
			To address these challenges, we present Robust Aggregation in Noise (RAIN), a unified framework that reconciles privacy, robustness, and verifiability under Shuffle-DP. At its core, RAIN adopts sign-space aggregation to robustly measure update consistency and limit malicious influence under noise and anonymization.
			Specifically, we design two novel secret-shared protocols for shuffling and
			aggregation that operate directly on additive shares and preserve Shuffle-DP’s
			tight privacy guarantee. In each round, the aggregated result is verified to ensure correct aggregation and detect any selective dropping, achieving malicious security with minimal overhead.
			Extensive experiments across comprehensive benchmarks show that RAIN maintains strong privacy guarantees under Shuffle-DP and remains robust to poisoning attacks with negligible degradation in accuracy and convergence. It further provides real-time integrity verification with complete tampering detection, while achieving up to 90× lower communication cost and 10× faster aggregation compared with prior work.
		\end{abstract}
		

		% \fi
		\IEEEpeerreviewmaketitle

		\section{Introduction}
		Distributed learning has emerged as a powerful paradigm for collaboratively training models across multiple data holders while keeping raw data local. Federated learning (FL) enables multiple clients to jointly train
		models by exchanging gradient or parameter updates with a central
		server instead of sharing raw data~\cite{r21}.
		However, subsequent studies have shown that such updates can still leak
		sensitive information through gradient inversion~\cite{r2} and membership
		inference attacks~\cite{r37}, and are vulnerable to poisoning by
		malicious clients~\cite{r50}.

		Differential privacy (DP)~\cite{r38} is widely adopted to mitigate
		such privacy leakage.
		Early approaches primarily employed DP in a centralized setting \cite{r39}, which relies on a fully trusted server to aggregate updates and inject noise. In contrast, local differential privacy (LDP) \cite{r40} has been recognized as more suitable for decentralized or cross-device learning scenarios \cite{r41}. Although LDP eliminates the need for a trusted aggregator, its strict privacy guarantees often require large noise~\cite{r42}, significantly degrading model utility and slowing convergence.
		
		Motivated by the need for strong privacy with practical utility, the shuffle
		model of differential privacy (Shuffle-DP)~\cite{r43} has been
		introduced and extended to distributed learning scenarios~\cite{r44}.
		By combining local perturbation with shuffling-based
		anonymization, Shuffle-DP amplifies privacy guarantees while
		requiring substantially less noise than LDP. As a result, Shuffle-DP has emerged as a prominent framework for privacy-preserving federated learning.
		
		Nevertheless, incorporating Shuffle-DP into FL introduces new security gaps.
		Noise and anonymization can erase the signals that robust aggregation relies on,
		leaving Shuffle-DP pipelines vulnerable to adversarial manipulation.
		These gaps manifest along two dimensions.
		
		\begin{figure}[t]
			\centering
			\includegraphics[width=1\linewidth]{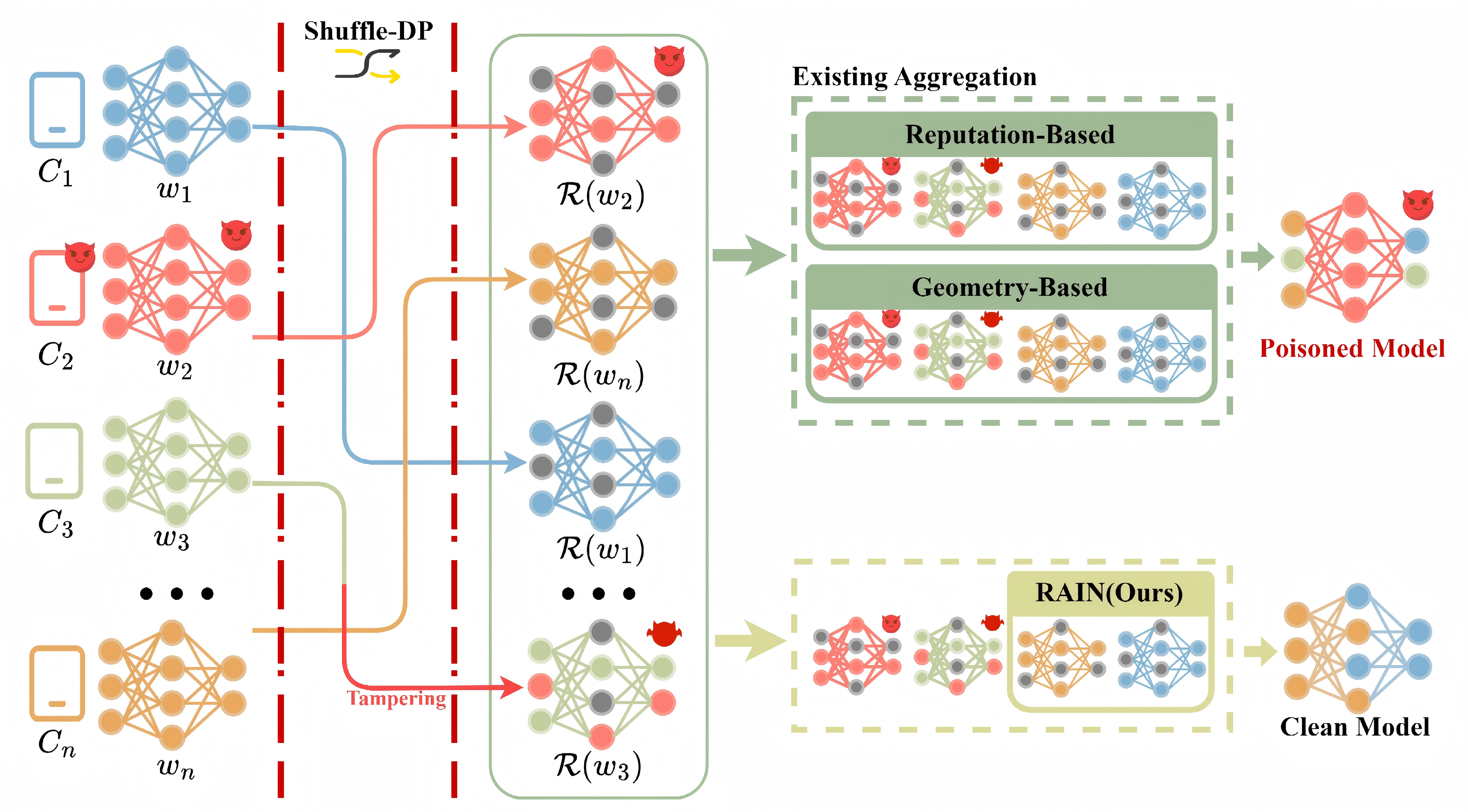}
			\caption{\textbf{Effect of Shuffle-DP on aggregation robustness.}
				Existing robust aggregation fails under Shuffle-DP:
				local noise distorts gradient geometry, anonymization removes
				identity continuity, and they are further vulnerable to tampering by the shuffler.
				RAIN preserves robustness under the same Shuffle-DP setting without
				relying on a trusted shuffler.
			}
			\label{fig_1}
		\end{figure}
		
		\textbf{1) Shuffle-DP Breaks Robust Aggregation Assumptions.}
		As shown in Fig.~\ref{fig_1}, conventional robust aggregation
		methods fail to remain effective under the Shuffle-DP setting. These approaches can be broadly categorized
		into geometry-based and reputation-based paradigms~\cite{r3}.
		Geometry-based schemes~\cite{r26,r31,r22} detect and filter malicious
		updates by exploiting geometric consistency among client gradients,
		such as distance or directional similarity.
		Reputation-based schemes~\cite{r24,r25,r23} rely on temporal
		consistency across training rounds, maintaining per-client trust
		scores or historical reputations to weight future updates.
		However, these principles no longer hold once Shuffle-DP is applied. This phenomenon is consistently observed in our experimental
		evaluation under Shuffle-DP.
		Geometric relationships among client updates are distorted, and
		identity continuity is removed under Shuffle-DP.
		Consequently, trustworthy clients could no longer be distinguished by the server.
		\textbf{Shuffle-DP breaks the assumptions underlying both
			geometry-based and reputation-based aggregation strategies,
			rendering them unreliable in adversarial settings.}
		
		\textbf{2) Untrusted Shuffler Can Tamper with Client Updates.}
		Existing Shuffle-DP frameworks~\cite{r17,r13} also exhibit fundamental limitations in their threat modeling.
		Nearly all prior studies assume that the shuffler behaves faithfully permuting the perturbed gradients received from clients and forwarding them to the server for aggregation.
		This strong trust assumption simplifies the protocol design but introduces a critical blind spot: if the shuffler is compromised or colludes with the server, it can tamper with, discard, or selectively reorder client updates during the shuffle process.
		Crucially, existing Shuffle-DP designs provide no mechanism to verify the correctness of the shuffle itself. As a result, \textbf{an untrusted shuffler can tamper with client updates
			without detection, undermining aggregation integrity and the
			privacy guarantees promised by Shuffle-DP.}
		
		To address the breakdown of robust aggregation under Shuffle-DP, we present Robust Aggregation in Noise (RAIN).The key insight behind RAIN is that,
		while Shuffle-DP anonymization distorts gradient magnitudes and geometry,
		the sign patterns of benign and malicious updates remain statistically separable.
		RAIN leverages this property to perform robust aggregation in the sign space,
		replacing geometry-dependent similarity with identity-free trust scoring
		to bound malicious influence under noise and anonymization.
		In addition, RAIN employs a secure secret-shared shuffle and aggregation protocol that executes all operations over additive shares, preserving privacy amplification without relying on a trusted shuffler.
		It further incorporates a lightweight streaming integrity layer
		with message authentication codes (MACs) to detect selective-drop
		and tampering attacks in real time, thereby ensuring that all
		client updates are included exactly once and remain unmodified,
		achieving verifiable aggregation integrity with minimal overhead.
		Taken together, these components form a unified Shuffle-DP
		framework that reconciles privacy preservation, robustness, and verifiable integrity.
		By operating on binary sign updates, RAIN further reduces
		communication cost, demonstrating that strong security and
		efficiency can coexist in federated learning.
		
		We implement and evaluate RAIN on representative benchmarks.
		Results show that RAIN consistently outperforms existing Shuffle-DP frameworks across model utility, robustness, and communication efficiency dimensions.
		For model utility, RAIN achieves stronger convergence and better privacy–utility trade-offs. For example, under $\varepsilon = 10$, RAIN attains 87.6\% accuracy on MNIST compared with 80.6\% for the state-of-the-art~\cite{r33}. Regarding robustness, RAIN maintains over 80\% accuracy even when the majority of clients are malicious, indicating strong resilience in malicious environments. Meanwhile, the integrity verification module achieves 100\% tampering detection with negligible runtime overhead,
		while the secret-shared shuffle naturally yields up to $90\times$ lower communication cost and $10\times$ faster runtime compared with prior Shuffle-DP frameworks.
		
		We highlight our main contributions below:
		\begin{itemize}
			\item
			We propose Robust Aggregation in Noise (RAIN), a sign-space aggregation mechanism, which quantifies updates consistency in a noise-resilient manner, bounding malicious influence even under anonymized and noisy updates.
			
			\item
			Building upon this mechanism, We design two novel secret-shared protocols for shuffling and aggregation that operate directly on additive shares and preserve Shuffle-DP’s tight privacy guarantees, enabling robust aggregation without relying on a trusted shuffler.
			
			\item
			We design a lightweight MAC-based layer for real-time verification during shuffling and aggregation, detecting message dropping and tampering attacks with minimal overhead while maintaining DP guarantees.
			
			\item
			We conduct extensive experiments on datasets including MNIST, Fashion-MNIST, and CIFAR-10, demonstrating that RAIN outperforms existing Shuffle-DP frameworks in privacy-utility trade-offs, robustness under adversarial settings, and communication efficiency.
		\end{itemize}

		\section{Related Work}
		
		\subsection{Privacy-Preserving Federated Learning}
		
		While robust aggregation seeks to ensure reliability under malicious updates, privacy-preserving federated learning focuses on protecting clients’ data from inference and exposure.
		A broad spectrum of privacy mechanisms has been developed to safeguard the information contained in gradient updates, among which differential privacy and secure multi-party computation (MPC) \cite{r35} are the most widely adopted.
		
		Early works applied central differential privacy (CDP) \cite{r39}, where a trusted server perturbs the aggregated model updates before dissemination.
		Although CDP provides formal privacy guarantees, it assumes that the server behaves honestly, a strong assumption rarely satisfied in practice.
		To eliminate this trust assumption, local differential privacy \cite{r40} was proposed, letting each client add random noise to its local gradients before transmission.
		LDP guarantees client-level privacy even against an untrusted aggregator, but the heavy noise required under strict privacy budgets severely hurts model utility and slows convergence~\cite{r34}.
		
		To alleviate this privacy–utility tension, recent work has adopted the shuffle model of differential privacy \cite{r43}. In Shuffle-DP, clients first locally perturb their updates and then submit them to a shuffler that randomly permutes and anonymizes messages before aggregation.
		This anonymization yields privacy amplification by shuffling, allowing the same formal guarantee with substantially less noise than in LDP.
		Shuffle-DP has thus become an attractive paradigm for scalable privacy-preserving FL and inspired frameworks such as FLAME \cite{r47}, and Camel \cite{r33}, which integrate DP with lightweight encryption or secure aggregation to improve communication efficiency.
		However, these systems focus mainly on confidentiality, assuming benign clients and an honest shuffler, and therefore remain vulnerable to model-poisoning or integrity attacks.
		
		Complementary to DP-based mechanisms, secure computation techniques such as MPC, homomorphic encryption (HE), and secure aggregation have been extensively studied to prevent raw gradient disclosure \cite{r11,r10}.
		These cryptographic protocols guarantee that the server learns only aggregated results but cannot access individual updates.
		Recent examples include FLOD and ELSA \cite{r6}, which employ MPC for secure model aggregation under honest-but-curious assumptions.
		While such schemes provide strong privacy guarantees, they typically incur high communication overhead and lack robustness to malicious manipulation within the encrypted domain.
		
		As summarized in Table \ref{tab:fl-comparison}, existing privacy-preserving FL frameworks address privacy or robustness only in isolation.
		RAIN unifies the strengths of both MPC and Shuffle-DP, achieving strong model privacy, robust aggregation, and high system efficiency under adversarial conditions.
		It remains effective even with malicious servers and provides verifiable, privacy-preserving aggregation without compromising accuracy or communication cost.

		\newcolumntype{C}[1]{>{\centering\arraybackslash}m{#1}}
		
		\begin{table}[!t]
			\centering
			\caption{Positioning RAIN Among Existing Approaches.}
			\label{tab:fl-comparison}
			\footnotesize
			\setlength{\tabcolsep}{0.5pt}
			\begin{threeparttable}
					\begin{tabular}{
							C{1.6cm}  % Solution
							C{1.5cm}  % Privacy Mechanism
							C{1.6cm}  % Robust Aggregation
							C{1.6cm}  % Malicious Security
							C{1.3cm}  % Verifiability
							C{0.9cm}  % Efficiency
						}
						\toprule
						\textbf{Proposed Works} &
						\textbf{Privacy Mechanism} &
						\textbf{Privacy Preservation} &
						\textbf{Robust Aggregation} &
						\textbf{Integrity} &
						\textbf{Efficiency} \\
						\midrule
						FedAvg \cite{r21}   & --       & \xmark & \xmark  & \xmark & \xmark \\
						FLguard \cite{r8}   & MPC      & \cmark & \cmark  & \xmark & \xmark \\
						FLOD \cite{r9}      & MPC      & \cmark & \qmark  & \xmark & \xmark \\
						ELSA \cite{r4}      & MPC      & \cmark & \qmark  & \cmark & \cmark \\
						FLAME \cite{r22}    & SDP      & \qmark & \cmark  & \qmark & \xmark \\
						Camel \cite{r33}    & MPC+SDP  & \cmark & \xmark  & \cmark & \qmark \\
						%				\addlinespace[2pt]
						%				\hline
						\addlinespace[2pt]
						\textbf{RAIN(Ours)}       & MPC+SDP   & \cmark & \cmark & \cmark & \cmark \\
						\bottomrule
					\end{tabular}%
				
				\vspace{3pt}
				\begin{flushleft}
					\small
					\cmark fully addressed, \qmark partially addressed, \xmark not addressed.
				\end{flushleft}
			\end{threeparttable}
		\end{table}

		\subsection{Robust Aggregation in Federated Learning}
		
		Federated learning is inherently vulnerable to poisoning and Byzantine attacks, where adversarial clients upload manipulated updates to corrupt the global model \cite{r45}.
		Because the local training process is invisible to the server, malicious participants can craft arbitrary gradients to degrade accuracy or embed hidden behaviors \cite{r46}.
		These threats have motivated a large body of research on robust aggregation, which aims to maintain global model reliability even when a fraction of clients behave adversarially.
		
		Early defenses rely on geometric consistency among client updates.
		Methods such as Krum, Trimmed-Mean, and Median~\cite{r31} treat malicious updates as outliers in Euclidean or coordinate space and aggregate only the statistically central gradients~\cite{r29}.
		While simple and effective under homogeneous data, these methods break down when benign updates become diverse or when privacy noise perturbs the geometric structure.
		Later, direction-aware approaches like FLTrust and FLAME exploit cosine similarity between client and reference models to detect anomalies, whereas FLARE and FLOD \cite{r9} project gradients into representation or sign space to improve resilience.
		Despite their sophistication, these techniques remain geometry-based, presuming visible, comparable gradients and a trustworthy aggregation channel~\cite{r27}.
		
		Complementary to geometric defenses, reputation-based aggregation builds temporal trust across training rounds.
		FoolsGold penalizes clients that exhibit correlated update directions to mitigate Sybil attacks,
		while DnC \cite{r20} and CONTRA learn dynamic trust weights or contrastive embeddings to identify persistent malicious behavior.
		These schemes shift from per-round geometry to historical behavioral consistency, improving long-term robustness but still assuming stable, identifiable clients.
		Once gradients become anonymized or shuffled, as in privacy-preserving settings these assumptions collapsed, leaving reputation-based defenses ineffective.
		
		Recent advances attempt to bridge robustness with privacy and efficiency\cite{r28}.
		FLOD leverages Hamming-distance–based sign aggregation under multi-server MPC to tolerate over 50\% Byzantine clients with encrypted updates, and FLARE utilizes feature-space detection for stealthy backdoors.
		However, these designs either impose high communication and computation costs or rely on partial trust assumptions, limiting scalability.
		
		Overall, existing robust aggregation schemes fundamentally rely on two forms of consistency—geometric alignment among visible gradients and temporal identity across training rounds.
		None remain effective once Shuffle-DP anonymization simultaneously destroys gradient geometry and identity continuity, exposing a fundamental incompatibility between privacy and robustness.
		
		RAIN addresses this gap by operating directly in the sign space, where update consistency is measured through Hamming distance rather than geometric similarity.
		This geometry-free design preserves robustness even under strong perturbation and anonymity, while its secure two-server realization ensures verifiable aggregation without relying on a trusted shuffler.

		\section{Preliminaries}
		
		\subsection{Differential Privacy}
		
		This section summarizes the basic definitions of differential privacy relevant to our setting. 
		In the central model of DP, a trusted curator collects users’ data and applies a randomized mechanism before releasing any aggregate statistics.  
		Formally, two datasets $\mathcal{D} = \{d_1, \ldots, d_n\}$ and $\mathcal{D}' = \{d'_1, \ldots, d'_n\}$ drawn from the same domain $\mathcal{X}$ are said to be neighboring if they differ in only one record, i.e., there exists an index $i$ such that $d_i \neq d'_i$ while $d_j = d'_j$ for all $j \neq i$.  
		A randomized mechanism $\mathcal{M}$ is differentially private if its outputs on neighboring datasets remain statistically close.
		
		\begin{definition}[\textbf{Central Differential Privacy}]
			A randomized mechanism $\mathcal{M}: \mathcal{X}^n \rightarrow \mathcal{Y}$ satisfies \textbf{$(\varepsilon, \delta)$-DP} if, for all neighboring datasets $\mathcal{D}, \mathcal{D}' \in \mathcal{X}^n$ and every subset $S \subseteq \mathcal{Y}$,
			\begin{equation}
				\Pr[\mathcal{M}(\mathcal{D}) \in S] \le e^{\varepsilon}\Pr[\mathcal{M}(\mathcal{D}') \in S] + \delta.
			\end{equation}
		\end{definition}
		
		In contrast, the local model removes the need for a trusted server by letting each user perturb its data locally before sending it for aggregation~\cite{r49}.  
		The following definition formalizes the local privacy guarantee.
		
		\begin{definition}[\textbf{Local Differential Privacy}]
			A mechanism $\mathcal{R}: \mathcal{X} \rightarrow \mathcal{Y}$ satisfies \textbf{$\varepsilon_0$-LDP} if, for any two possible inputs $d, d' \in \mathcal{X}$ and any subset $S \subseteq \mathcal{Y}$,
			\begin{equation}
				\Pr[\mathcal{R}(d) \in S] \le e^{\varepsilon_0}\Pr[\mathcal{R}(d') \in S].
			\end{equation}
		\end{definition}
		
		While LDP eliminates the need for trust, achieving strict local privacy often requires large noise, which significantly harms model utility.  
		
		\begin{definition}[\textbf{Shuffled Differential Privacy}]
			A protocol consisting of a local randomizer $\mathcal{R}$, a shuffler
			$\mathcal{S}$, and an analyzer $\mathcal{A}$ satisfies
			\textbf{$(\varepsilon, \delta)$-Shuffle-DP} if the overall mechanism
			$\mathcal{A} \circ \mathcal{S} \circ \mathcal{R}^n$ is
			$(\varepsilon, \delta)$-differentially private in the central model.
		\end{definition}
		The shuffle model bridges the gap between the central and local models by
		introducing an anonymization step between local randomization and aggregation.
		In this model, each client first applies a local randomizer and then sends the
		perturbed output to a shuffler, which randomly permutes all messages before
		forwarding them to the analyzer.
		By breaking the linkage between users and their messages, shuffling enables
		significant privacy amplification compared to pure local DP.
		
		RAIN adopts the Gaussian mechanism as the local randomizer, where each client
		clips its gradient and adds Gaussian noise $\eta_i \sim \mathcal{N}(0, \sigma^2 I_d)$.
		When combined with shuffling, the resulting protocol satisfies
		$(\varepsilon, \delta)$-Shuffle-DP with significantly amplified privacy compared
		to pure LDP, providing the formal basis for client-level privacy in our setting.

		\subsection{Cryptographic Foundations}
		
		\textbf{Secure Two-Party Computation (2PC).}
		RAIN relies on a lightweight two-server protocol based on additive secret sharing,
		which enables two non-colluding servers to jointly compute functions over private data
		without revealing individual inputs.
		
		\textbf{Additive Secret Sharing.}
		For a value $x \in \mathbb{Z}_p$, client $C_i$ generates two random shares
		$(x_0, x_1)$ such that $x = x_0 + x_1 \bmod p$.
		Server $S_t$ ($t \in \{0,1\}$) holds share $x_t$.
		Neither server can reconstruct $x$ alone, but they can jointly perform secure
		additions and multiplications over the shares.
		
		\textbf{Secure Multiplication via Beaver Triples.}
		Given shared values $\langle a \rangle^{\mathbb{A}}$ and $\langle b \rangle^{\mathbb{A}}$,
		the parties use a pre-generated Beaver triple
		$(x, y, z)$ satisfying $z = xy$ to compute $ab$ securely:
		\begin{equation}
			\langle ab \rangle_t^{\mathbb{A}}
			= -t e f + f \langle a \rangle_t^{\mathbb{A}}
			+ e \langle b \rangle_t^{\mathbb{A}}
			+ \langle z \rangle_t^{\mathbb{A}},
		\end{equation}
		where $e = a - x$ and $f = b - y$ are reconstructed
		after local subtraction of shares.
		This primitive allows all subsequent arithmetic in RAIN—including
		Hamming-distance evaluation, ReLU-based weighting, and aggregation—to be executed
		entirely over secret shares.
		
		For efficiency, RAIN employs only arithmetic sharing over $\mathbb{Z}_p$;
		Boolean sharing and garbled circuits are omitted,
		as no bit-level logic is required in our design.

		\section{Problem Statement}
		
		\subsection{System Model}
		Figure~\ref{fig_2} illustrates the federated learning setting considered in this work. 
		We assume a population of $N$ clients $\mathcal{C}=\{C_1,\ldots,C_N\}$ and two aggregation servers $S_0$ and $S_1$ that jointly provide a secure shuffling-and-aggregation service. 
		Each client $C_i$ holds a private dataset $D_i$ and, at round $t$, performs local training to obtain a gradient update $g_i^t=\nabla\mathcal{L}_i(w_t)$. 
		To provide local randomization compatible with the shuffle model of DP, the client clips its update, adds Gaussian noise, and applies sign encoding, yielding a noisy binary vector $g_i^{\mathrm{bool}}$.
		
		All noisy sign updates are transmitted to the two servers over secure channels. 
		The servers execute a secure shuffling procedure that permutes the incoming updates to break linkability between clients and messages. 
		After shuffling, they perform integrity checks and run a privacy-preserving aggregation over the anonymized, noised updates. 
		The global model is then updated as
		\begin{equation}
			\label{eq:update_aggregate}
			\tilde{W} = w_t - \eta \cdot \textsc{Aggregate}\big(\{g_i^{\mathrm{bool}}\}\big),
		\end{equation}
		where $\textsc{Aggregate}(\cdot)$ denotes an aggregation mechanism that operates under the above privacy constraints. 
		Communication between clients and servers is protected using secure multi-party computation primitives, so that no single server can reconstruct any individual update or identify its origin.
		
			\begin{figure*}[t]
			\centering
			\includegraphics[width=1.05\textwidth]{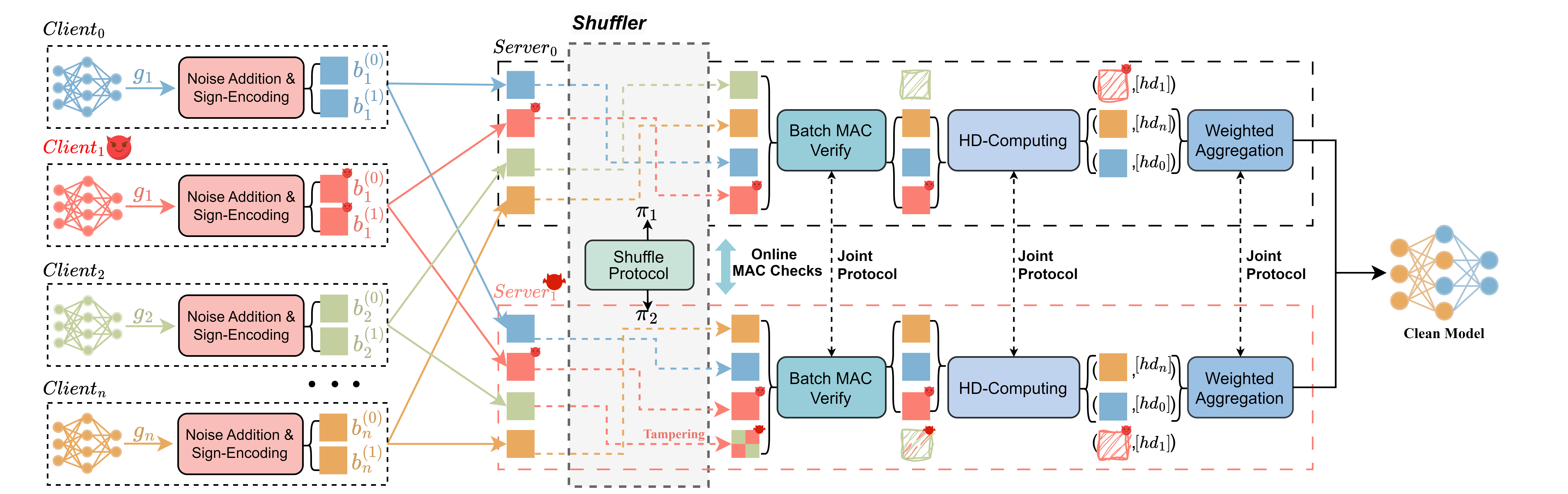}
			\caption{
				\textbf{Overview of the RAIN framework.}
				Each client $C_i$ locally adds DP noise and sign-encodes its update into a binary vector, then sends secret shares to two non-colluding servers $S_0$ and $S_1$.
				The servers execute a verifiable shuffle to anonymize message origins, perform MAC-based integrity checks to reject dropped or tampered messages, and compute Hamming-distance consistency over the shuffled sign vectors.
				RAIN then aggregates the verified updates in the arithmetic domain and reconstructs the global model $\tilde{\mathbf{W}}$.}
			\label{fig_2}
		\end{figure*}
		
		Certain robust aggregation mechanisms rely on a public reference signal to stabilize aggregation under anonymized and noisy updates.
		We assume access to a small, audited public calibration set $\mathcal{R}$ (size $\approx 10^2$), drawn to match the task distribution and containing no client data.
		$\mathcal{R}$ is used only to compute a stable reference direction and never participates in training or privacy accounting.
		The set can be curated by a jointly trusted authority or sourced from public benchmarks; its small size makes auditing feasible. We treat the resulting reference direction as public and observable to adversaries. Even under this assumption, an adversary can only affect aggregation by flipping
		a bounded fraction of sign coordinates without exceeding the Hamming-distance threshold.
		This assumption is standard and does not weaken client-side Shuffle-DP guarantees.

		\subsection{Threat Model and Assumptions}
		We consider an adversary with both internal and external capabilities, targeting robustness, privacy, and integrity of the federated learning process. 
		All parties are assumed to be computationally bounded and cannot break standard cryptographic primitives such as MPC, message authentication codes, or encryption.
		
		\textbf{Client-side adversary.}
		The adversary may compromise an arbitrary, potentially majority, subset of clients and coordinate them to mount Byzantine attacks that bias the global model or degrade its accuracy. 
		Compromised clients can arbitrarily manipulate local gradients before randomization and may attempt to infer information about benign clients from aggregate updates or model dynamics.
		%RAIN’s robust aggregation design aims to bound such adversarial influence even when most participating clients are malicious.

		\textbf{Server-side adversary.}
		%Among the two aggregation servers $(S_0, S_1)$, we adopt a non-colluding two-server model: at most one server may be malicious or semi-honest.  
		%A malicious server may attempt to infer client participation, reorder or modify uploaded updates, or selectively drop messages during shuffling and aggregation.  
		%A semi-honest server follows the protocol but tries to learn additional information from intermediate states.  
		%This assumption is standard in privacy-preserving analytics (e.g., SecureFL~\cite{r7}, ELSA~\cite{r6}, FLOD~\cite{r9}) and is practical in cross-organizational deployments—such as a cloud provider cooperating with an independent auditor or two departments governed by distinct administrative domains.  
		%If both servers collude, privacy degrades to the local model, where protection falls back to LDP-level guarantees.
		We consider adversaries that may compromise the server-side infrastructure
		responsible for shuffling and aggregation.
		To model this setting, we adopt a non-colluding two-server architecture
		with two aggregation servers $(S_0, S_1)$,
		where at most one server may be malicious or semi-honest.
		A compromised server may deviate arbitrarily from the protocol execution,
		while the other remains honest.
		This non-collusion assumption is standard in privacy-preserving analytics
		(e.g., SecureFL~\cite{r7}, ELSA~\cite{r6}, FLOD~\cite{r9})
		and reflects practical cross-organizational deployments,
		such as a cloud provider cooperating with an independent auditor
		or two entities governed by distinct administrative domains.
		If both servers collude, privacy degrades to the local model,
		where protection reduces to standard LDP guarantees.

		%\textbf{Integrity threats.}
		%During shuffling and aggregation, an active adversary may attempt to tamper with, replay, or selectively drop client messages to bias the aggregated result or disrupt system operation.  
		%To defend against such threats, RAIN integrates batch and online MAC verification, enabling real-time detection of message modification or loss while preserving anonymity.

		\section{The Design of RAIN}
		
		\subsection{Overview}
		
		RAIN aims to enable secure and robust federated learning under the Shuffle-DP model,
		where client-level privacy, Byzantine robustness, and verifiable aggregation integrity
		must be achieved simultaneously.
		Unlike prior Shuffle-DP frameworks that address privacy and robustness in isolation,
		RAIN unifies these objectives through a two-layer design:
		the RAIN mechanism, a sign-space robust aggregation rule,
		and the RAIN framework, its malicious-secure realization via
		a secure secret-shared shuffle and aggregation protocol with integrated integrity verification.
		
		The key challenge for achieving robustness under Shuffle-DP
		is that random permutation, while amplifying privacy,
		simultaneously destroys the geometric consistency and temporal identity
		that conventional robust aggregation depends on.
		As a result, geometric-based and reputation-based defenses collapse under anonymization.
		To restore robustness without breaking anonymity,
		we propose Robust Aggregation in Noise,
		which aggregates in the sign domain instead of the Euclidean space.
		By exploiting the statistical separability between benign and adversarial sign vectors,
		RAIN replaces geometric similarity with Hamming-distance–based trust scoring,
		thus bounding the influence of malicious clients even when a majority of participants are adversarial.
		This RAIN mechanism provides the theoretical foundation of our framework.

		While sign-space aggregation enables robustness under noise and anonymity,
		implementing it securely under Shuffle-DP raises additional challenges in privacy.
		A naive approach would be to adopt verifiable mixnet-style shuffling or generic secure aggregation protocols.
		However, these methods rely on public-key operations and zero-knowledge proofs,
		resulting in prohibitive computational and communication cost
		when processing high-dimensional model updates in federated learning.
		Moreover, mixnet-based verifiability assumes a trusted auditor or non-overlapping batches,
		which conflicts with the continuous aggregation requirement of FL. To overcome these limitations, the RAIN framework
		builds upon lightweight secret sharing and modular arithmetic
		to realize a two-server secret-shared shuffle and aggregation protocol.
		Two non-colluding servers jointly perform all shuffle and aggregation steps
		directly over additive shares of sign-encoded updates,
		without learning either the client identities or the shuffle permutation.
		This preserves Shuffle-DP privacy amplification
		while ensuring computational correctness even if one server is malicious or semi-honest.

		Beyond robustness and privacy, RAIN further guarantees verifiable integrity
		against selective-failure and tampering attacks from a malicious server.
		Prior MPC shuffle frameworks achieve verifiability
		through heavy cryptographic proofs or offline auditing,
		which are infeasible for real-time federated aggregation.
		We instead design a streaming integrity verification layer
		that attaches lightweight message authentication codes
		to each encrypted update and validates them online throughout the shuffle pipeline.
		This streaming design ensures that every client contribution is either correctly processed or detected as invalid in real time,
		without interrupting the MPC execution or violating DP semantics.
		Formally, we prove in Section~\ref{sec:integrity} that our protocol
		achieves malicious-security against any single colluding server,
		ensuring that no message can be dropped, modified, or replayed without detection.
		\subsection{Robust Aggregation in Noise}
		\label{sec:rain}
		
		Traditional robust aggregation schemes rely on the geometric or temporal consistency of client updates to identify and suppress malicious gradients.
		However, in the Shuffle-DP setting, the injected noise distorts geometric relations, and the shuffling process removes any persistent identity linkage.
		As a result, distance- or reputation-based aggregation rules fail to operate effectively.
		To overcome this limitation, we design RAIN,
		a geometry-free robust aggregation mechanism that operates directly in the sign space.
		
		Although noise perturbation and anonymization obscure gradient magnitudes and directions, 
		the element-wise sign pattern of updates remains statistically distinguishable between benign and adversarial clients.
		Benign updates tend to produce sign vectors aligned with the true global direction,
		whereas poisoned or random updates exhibit significantly higher sign-flip ratios.
		Leveraging this property, RAIN aggregates in the sign domain using Hamming-distance–based trust scoring and ReLU-adaptive weighting,
		thereby bounding the influence of outliers without requiring client identities or raw geometric information.

		\textbf{Local Randomization.}
		To ensure client-level privacy compatible with Shuffle-DP, we adopt the \emph{Sign-Gaussian mechanism}~\cite{r18}, 
		originally proposed in DP-RSA to jointly achieve differential privacy and Byzantine robustness. 
		Each client applies the mechanism to its clipped gradient:
		\begin{equation}
			\tilde{b}_i = \mathrm{sign}\!\bigl(\mathrm{clip}(g_i) + \eta_i\bigr), 
			\qquad \eta_i \sim \mathcal{N}(0,\sigma^2 I_d),
		\end{equation}
		which satisfies $(\varepsilon,0)$-DP when 
		$\sigma > \max_i\!\{\tfrac{2|g_i|}{3},\, \tfrac{4\Delta_g}{\varepsilon}\}$. 
		This perturbation injects randomness primarily near the sign boundary, 
		preserving most directional information while maintaining strict privacy under Shuffle-DP anonymization.
		
		Let $g_i$ denote a scalar component of the clipped gradient and $\eta_i \sim \mathcal{N}(0,\sigma^2)$ the added noise. 
		The transmitted sign is $\tilde{b}_i = \mathrm{sign}(g_i + \eta_i)$. 
		The probability that the noise flips the sign is
		\begin{equation}
			P_{\mathrm{flip}} = \Pr[\tilde{b}_i \neq \mathrm{sign}(g_i)] 
			= \Phi\!\left(-\frac{|g_i|}{\sigma}\right),
		\end{equation}
		where $\Phi(\cdot)$ is the standard Gaussian CDF. 
		Hence, larger $|g_i|$ or smaller $\sigma$ yields a lower flip probability. 
		The expected transmitted sign satisfies
		\begin{equation}
			\mathbb{E}[\tilde{b}_i \mid g_i]
			= (2\Phi(|g_i|/\sigma)-1)\,\mathrm{sign}(g_i),
		\end{equation}
		indicating that the noisy sign is an attenuated version of the true sign with factor 
		$\kappa_i = 2\Phi(|g_i|/\sigma)-1 \in [0,1]$. 
		Following the privacy analysis in~\cite{r18}, 
		the mechanism satisfies $(\varepsilon,0)$-DP when
		\begin{equation}
			\sigma > \max\!\left\{\frac{2|g_i|}{3},\ \frac{4\Delta_g}{\varepsilon}\right\}.
		\end{equation}
		This reveals the intrinsic privacy–robustness trade-off: 
		larger $\sigma$ enhances privacy but increases sign-flip probability, 
		while smaller $\sigma$ improves robustness at the cost of weaker privacy. 
		As training converges and $\|g_i\|$ decreases, the effective noise impact naturally diminishes, 
		maintaining a balanced privacy–utility trade-off.
		
		\begin{algorithm}[h]
			\small
			\caption{RAIN: Robust Aggregation in Noise}
			\label{alg:rain}
			\KwIn{Randomized sign vectors $\{\tilde{b}_i\}_{i=1}^N$, reference direction $r$ (from root dataset), MAD coefficient $\lambda_{\mathrm{mad}}$}
			\KwOut{Aggregated update $g_{\text{agg}}$}
			\For{each client $i$}{
				Compute $d_i = \tfrac{1}{2d}(d - \tilde{b}_i\!\cdot r)$\;
			}
			Estimate $\tau = \mathrm{median}(d) + c_{\mathrm{mad}}\,\cdot\lambda_{\mathrm{mad}}\cdot \mathrm{MAD}(d)$\;
			Assign $\alpha_i = \tfrac{\max(0,\,\tau-d_i)}{\sum_j \max(0,\,\tau-d_j)}$\;
			Compute $g_{\text{agg}} = \sum_i \alpha_i \tilde{b}_i$\;
			\Return $g_{\text{agg}}$\;
		\end{algorithm}

		\textbf{Aggregation Rule.}
		Let $\tilde{b}_i \in \{-1,+1\}^d$ denote the randomized sign vector from client~$i$.
		To anchor trust in an anonymized environment, 
		RAIN adopts a reference direction $r$ derived from a small trusted root dataset,
		which provides a stable and verifiable global orientation.
		The server (or secret-shared aggregator) computes the Hamming distance between each $\tilde{b}_i$ and $r$:
		\begin{equation}
			d_i = \frac{1}{2d}\bigl(d - \tilde{b}_i \cdot r\bigr).
		\end{equation}

		A robust threshold $\tau$ is estimated via the median and MAD:
		\begin{equation}
			\tau = \mathrm{median}(d) + c_{\mathrm{mad}}\,\cdot\lambda_{\mathrm{mad}}\cdot \mathrm{MAD}(d).
		\end{equation}
		Here, $c_{\mathrm{mad}} = 1.4826$ is the normal-consistency constant
		that scales the MAD to an unbiased estimate of the standard deviation
		under Gaussian assumptions.
		This normalization ensures that $\lambda_{\mathrm{mad}}$
		corresponds to the number of robust standard deviations from the median.

		Each client receives a ReLU-based trust weight:
		\begin{equation}
			\alpha_i = \frac{\max(0,\, \tau - d_i)}{\sum_j \max(0,\, \tau - d_j)},
		\end{equation}
		and the aggregated update is computed as
		\begin{equation}
			g_{\text{agg}} = \sum_{i=1}^{N} \alpha_i \tilde{b}_i, 
			\qquad
			\tilde{W} = w_t - \eta\, g_{\text{agg}}.
		\end{equation}

		The ReLU weighting acts as a smooth relaxation of binary trust filtering, 
		assigning higher weights to clients whose sign vectors align with the reference direction 
		and suppressing those deviating beyond the robust threshold. 
		Unlike FLTrust, FLARE, or Krum, RAIN operates purely in the sign domain, 
		requiring neither gradient magnitudes nor persistent client identities. 
		This makes RAIN inherently compatible with Shuffle-DP anonymization. 
		Moreover, since Hamming distance and ReLU can be expressed using addition and comparison primitives, 
		the entire aggregation can be securely executed within the MPC-encrypted domain without revealing individual updates. 
		The trusted root direction provides a consistent anchor, 
		ensuring stability even under heavy anonymization and noise perturbation.

		\textbf{Theoretical Guarantee.}
		Following~\cite{r18}, we derive the following bound on RAIN’s privacy–robustness trade-off.
		
		\begin{lemma}[Noise–Robustness Trade-off of RAIN]
			\label{lem:rain-tradeoff}
			Let $g_i$ be a scalar component of a clipped gradient with sensitivity $\Delta_g$, 
			and $\tilde{b}_i = \mathrm{sign}(g_i + \eta_i)$ its privatized sign under the Sign-Gaussian mechanism, 
			where $\eta_i \sim \mathcal{N}(0,\sigma^2)$.
			If $\sigma > \tfrac{4\Delta_g}{\varepsilon}$, the mechanism preserves $(\varepsilon,0)$-DP. 
			Moreover, when $\varepsilon \ge 2$, the expected sign-flip probability 
			$P_{\mathrm{flip}} = \Phi(-|g_i|/\sigma)$ 
			remains below $0.15$ for typical gradients, 
			and the aggregated direction in RAIN satisfies
			\[
			\mathbb{E}[\cos\theta] \ge 1 - 2P_{\mathrm{flip}} > 0.7.
			\]
			Hence, RAIN achieves both strict differential privacy and stable robust aggregation.This lemma analyzes the effect of Gaussian noise on sign stability
			and does not provide an additional privacy guarantee.
			
		\end{lemma}
		
		\begin{proof}
			From~\cite{r18}, the Sign-Gaussian mechanism guarantees $(\varepsilon,0)$-DP if $\sigma > 4\Delta_g/\varepsilon$.
			The sign-flip probability satisfies $P_{\mathrm{flip}} = \Phi(-|g_i|/\sigma) \le \Phi(-|g_i|\varepsilon/(4\Delta_g))$.
			For normalized gradients $|g_i|\ge 2\Delta_g$, we have $P_{\mathrm{flip}} \le \Phi(-\varepsilon/2)$.
			Numerically, $\Phi(-1)\!\approx\!0.1587$ and $\Phi(-1.5)\!\approx\!0.0668$, 
			so when $\varepsilon\!\ge\!2$, $P_{\mathrm{flip}}\!<\!0.15$.
			Let $p_f=P_{\mathrm{flip}}$.
			Since each flipped sign reverses one coordinate,
			the expected cosine similarity between RAIN’s aggregated direction and the true gradient satisfies $\mathbb{E}[\cos\theta] \ge 1-2p_f$.
			For $p_f\le0.15$, $\mathbb{E}[\cos\theta]\ge0.7$, ensuring convergence consistency.
			Thus, for $\varepsilon\ge2$, RAIN attains strict $(\varepsilon,0)$-DP 
			while maintaining robust aggregation with negligible degradation.
		\end{proof}

		\subsection{Secret-shared Shuffle and Aggregation}
		\label{sec:mpcshuffle}
		
		While RAIN enables robust aggregation in the sign domain, it still requires a secure mechanism to shuffle and aggregate anonymized updates without revealing their correspondence to individual clients.  
		In the Shuffle-DP paradigm, the shuffle operation itself provides privacy amplification, but realizing it securely in a multi-server environment is non-trivial: a curious or malicious server may attempt to infer the permutation, drop or modify messages, or bias the aggregated result.  
		To address these issues, RAIN introduces a lightweight two-server secret-shared shuffle and aggregation protocol, which performs both anonymization and robust aggregation entirely within the encrypted domain.
		
		\begin{algorithm}[t]
			\small
			\caption{MPC-based Encrypted Shuffle and Aggregation in RAIN}
			\label{alg:rain-ssa}
			\DontPrintSemicolon
			\SetKwInput{KwIn}{Input}\SetKwInput{KwOut}{Output}\SetKwInput{KwPar}{Parameters}
			\KwIn{Clients' additive shares $\{(b_i^{(0)},b_i^{(1)})\}_{i=1}^K$, where $b_i\!\in\!\{-1,+1\}^d$ and $b_i^{(0)}\!+\!b_i^{(1)}\!=\!b_i \bmod p$; \\
				\hspace{2.2em} Trusted direction share $[s_{\mathrm{trust}}]$;\\ \hspace{2.8em}robust threshold share $[\tau]$.}
			\KwOut{Aggregated next direction $[s_{\mathrm{next}}]\in\{-1,+1\}^d$.}
			\KwPar{Two non-colluding servers $S_0,S_1$; modulus $p$; PRG $G$; secure \textsf{Add}, \textsf{Mul}, \textsf{Cmp}.}
			
			\textbf{Offline (preprocessing)}\;
			\Indp
			\textbf{(O1)} $S_0$ and $S_1$ agree on a shared seed $\mathtt{seed}$; expand via PRG: 
			$(\pi_0,\pi_1,\{r_0[i]\}_{i=1}^K,\{r_1[i]\}_{i=1}^K)\leftarrow G(\mathtt{seed})$, with $r_0[i]+r_1[i]=0 \bmod p$.\;
			\textbf{(O2)} Fix MPC primitives for vector \textsf{Mul} (scalar$\times$vector), \textsf{Cmp} (comparison), and batched reductions.\;
			\Indm
			
			\textbf{Phase~I: Secret-shared Shuffle}\;
			\Indp
			\For{$i\leftarrow 1$ \KwTo $K$}{
				$S_0$: $b_i'^{(1)} \leftarrow \pi_0(b_i^{(0)}) + r_0[i]$
				
				$S_1$: $b_i'^{(2)} \leftarrow \pi_1(b_i^{(1)}) + r_1[i]$
			}
			\Indm
			
			\textbf{Phase~II: Secret-shared Robust Aggregation}\;
			\Indp
			\textbf{(A1) Secure Hamming distances}\;
			\For{$i\leftarrow 1$ \KwTo $K$}{
				\For{$j\leftarrow 1$ \KwTo $d$}{
					\lIf(\tcp*[f]{$x\oplus y = x+y-2xy$ over $\mathbb{Z}_p$}){$\ $}{$[\mathsf{xor}_{i,j}] \leftarrow [b'_{i,j}] + [s_{\mathrm{trust},j}] - 2\cdot [b'_{i,j}]\cdot [s_{\mathrm{trust},j}]$}
				}
				$[hd_i] \leftarrow \sum_{j=1}^d [\mathsf{xor}_{i,j}]$
			}
			\textbf{(A2) Adaptive ReLU weights}\;
			\For{$i\leftarrow 1$ \KwTo $K$}{
				$[m_i] \leftarrow \textsf{Cmp}([\tau], [hd_i])$ 
				$[w_i] \leftarrow ([\tau]-[hd_i])\cdot [m_i]$ 
			}
			\textbf{(A3) Weighted aggregation in shares}\;
			$[z] \leftarrow \mathbf{0}\in \mathbb{Z}_p^d$\;
			\For{$i\leftarrow 1$ \KwTo $K$}{
				$[z] \leftarrow [z] + [w_i]\cdot [b'_i]$ 
			}
			\textbf{(A4) Sign extraction (coordinate-wise)}\;
			\For{$j\leftarrow 1$ \KwTo $d$}{
				$[c_j]\leftarrow \textsf{Cmp}([z_j], 0)$ 
				$[s_{\mathrm{next},j}] \leftarrow 2\cdot[c_j]-1$
			}
			\Indm
			
			\Return $[s_{\mathrm{next}}]$.
			
		\end{algorithm}
		
		\textbf{Secret Sharing Scheme.}  
		RAIN employs a two-party additive secret-sharing scheme between two non-colluding servers, $S_0$ and $S_1$.  
		Each client $C_i$ splits its randomized sign vector $b_i \in \{-1,+1\}^d$ into two additive shares $(b_i^{(0)}, b_i^{(1)})$ satisfying
		\begin{equation}
			b_i^{(0)} + b_i^{(1)} = b_i \pmod p,
		\end{equation}
		and sends $b_i^{(t)}$ to server $S_t$ for $t\in\{0,1\}$.  
		Neither server can reconstruct any $b_i$ individually.  
		This setting eliminates the need for a central coordinator or pre-shared permutation (as in prior three-server schemes~\cite{r33}), while maintaining the same privacy amplification guarantee as Shuffle-DP.
		
		\textbf{Secret-shared Shuffle Protocol.}
		We instantiate the encrypted shuffle in RAIN using a lightweight two-party secret-shared protocol built on additive secret sharing.
		Unlike the three-server shuffle adopted by Camel~\cite{r12}, our design requires only two non-colluding servers, $S_0$ and $S_1$, 
		which collaboratively perform a secret permutation and re-masking of client updates in the encrypted domain.

		Each client $i$ provides additive shares $(b_i^{(0)}, b_i^{(1)})$ such that 
		$b_i^{(0)} + b_i^{(1)} = b_i \bmod p$.
		To securely shuffle the shared updates $\{[b_i]\}$, the two servers jointly derive correlated randomness from a shared PRG seed. 
		Specifically, $S_0$ and $S_1$ locally expand the seed to generate:
		\begin{equation}
			\pi_0,\pi_1 : [N] \!\to\! [N], \qquad 
			r_0,r_1 \in \mathbb{Z}_p^{N\times d},
		\end{equation}
		where $\pi_0,\pi_1$ are random local permutations and $r_0,r_1$ are random re-masking vectors satisfying $r_0 + r_1 = 0 \pmod p$.
		The secret-shared shuffle proceeds in two phases as follows.
		
		\textbf{Offline Phase.}
		The servers synchronize the random seeds used for permutation and re-masking via a common setup channel, 
		ensuring that the generated correlations $(\pi_0,\pi_1,r_0,r_1)$ are consistent but remain unknown to any third party.
		No client or external coordinator is required.
		
		\textbf{Online Phase.}
		During the shuffle, each server permutes and re-masks its local shares independently:
		\begin{equation}
			\begin{aligned}
				S_0: &\quad b'_i = \pi_0(b_i^{(0)}) + r_0[i],\\
				S_1: &\quad b'_i = \pi_1(b_i^{(1)}) + r_1[i].
			\end{aligned}
		\end{equation}
		Because $r_0 + r_1 = 0$, the permuted shares satisfy:
		\begin{equation}
			b_{i}'^{(0)}  + b_{i}'^{(1)}  = b_{\pi(i)} \bmod p,
		\end{equation}
		where $\pi = \pi_1 \circ \pi_0$ is the effective joint permutation. 
		Thus, the output $[b'_i]$ corresponds to a randomly permuted and re-masked version of the input $[b_i]$,
		achieving both anonymity and correctness:
		\begin{equation}
			[b'_i] \leftarrow \textsc{Secret-Shared Shuffle}_\pi([b_i]).
		\end{equation}
		Neither server learns the global permutation $\pi$ nor any link between a client and its shuffled update, 
		yet both can continue performing secure computation over the anonymized shares.
		This shuffle step provides the same privacy amplification as the Shuffle-DP model while preserving full compatibility with RAIN’s robust aggregation pipeline.

		\textbf{Secret-shared Robust Aggregation.}
		After the encrypted shuffle, the two servers jointly perform RAIN’s sign-space robust aggregation entirely over secret shares.
		For each anonymized client update $[b'_i]$, the servers first compute the bitwise Hamming distance to the trusted reference direction $[s_{\text{trust}}]$
		and then obtain an adaptive weight using a secure ReLU gate:
		\begin{equation}
			\label{eq:hammingrelu}
			\begin{aligned}
				[hd_i] &= \sum_{j=1}^{d} \big([b'_{i,j}] + [s_{\text{trust},j}] - 2\,[b'_{i,j}]\,[s_{\text{trust},j}]\big),\\
				[w_i]  &= \big([\tau]-[hd_i]\big)\cdot \mathbf{1}_{[\tau]\ge [hd_i]}.
			\end{aligned}
		\end{equation}
		These two steps are realized via the MPC-Hamming and MPC-ReLU subroutines, implemented using only additions,
		multiplications, and secure comparisons over $\mathbb{Z}_p$.
		
		Each weighted update is then aggregated in the encrypted domain through secure scalar–vector multiplication followed by coordinate-wise summation:
		\begin{equation}
			[z] = \sum_{i=1}^{K} [w_i]\,[b'_i].
		\end{equation}
		Finally, the servers execute a coordinate-wise MPC-Sign to derive the next trusted direction.
		All arithmetic and comparison primitives are linear in both the number of clients $K$ and the model dimension $d$,
		so the entire aggregation achieves $O(Kd)$ computational and communication complexity.
		
		Compared with mixnet-style or three-server verifiable shuffles,
		RAIN’s two-server protocol performs all robust aggregation steps, including Hamming distance computation, ReLU weighting, and sign extraction, directly over additive secret shares without relying on any public-key or zero-knowledge operations.
		This greatly reduces both computation and bandwidth.
		All randomness comes from symmetric-key PRG expansion, and every MPC primitive involves only linear rounds of interaction.
		Consequently, the protocol preserves the privacy amplification effect of Shuffle-DP,
		while securely realizing geometry-free robust aggregation inside the encrypted domain.
		
		By tightly integrating encrypted shuffling with robust sign-space aggregation,
		RAIN provides a practical and scalable realization of end-to-end privacy, correctness, and robustness
		within the Shuffle-DP federated learning paradigm using only two servers.

		\subsection{Streaming Integrity Verification}
		\label{sec:integrity}
		
		Existing approaches for malicious-secure shuffling, such as verifiable mixnets~\cite{r12} 
		and three-party ZK-based protocols, achieve correctness through heavy public-key proofs or batch audit phases.
		While these designs guarantee integrity, they incur prohibitive costs for high-dimensional FL updates, 
		requiring multiple encryption rounds, NIZK proofs, or offline audits for each permutation.
		Moreover, post-shuffle verification enables selective-failure attacks, 
		where a malicious server perturbs a small subset of shuffled entries and infers sensitive information 
		from which verification checks fail. 
		Hence, existing solutions cannot ensure real-time malicious security in practical federated settings.
		
		To overcome these limitations, RAIN introduces a streaming integrity verification layer that achieves continuous, online detection of integrity violations within the shuffle pipeline, 
		without revealing any permutation or message content.
		
		\textbf{Streaming MAC Construction.}
		Our protocol adopts a Carter Wegman symmetric message authentication code~\cite{r47}, 
		modified to support one-time authentication over secret-shared data streams.  
		Let $[r] = ([r]^{(0)}, [r]^{(1)})$ denote an additive share of a masked vector in $\mathbb{Z}_p^d$.  
		For a PRG-derived MAC key $k_t \in \mathbb{Z}_p^d$ generated from a time-dependent seed $\textit{seed}_t$, 
		the streaming MAC tag is computed as:
		\begin{equation}
			\label{eq:mac}
			t = \sum_{j=1}^{d} k_t[j] \cdot r[j].
		\end{equation}
		Both servers expand the shared seed using a synchronized PRG $G$, yielding additive key shares 
		$k_t^{(0)} + k_t^{(1)} = k_t \bmod p$, so that neither server learns the complete key.  
		At each shuffle iteration $t$, both servers compute and exchange the tag of each re-masked share:
		\begin{equation}
			\mathrm{tag}_i^{(j)} = \mathrm{MAC}_{k_t^{(j)}}(b_i'^{(j)}),
		\end{equation}
		and verify upon receipt:
		\begin{equation}
			\mathrm{Verify}\big(\mathrm{MAC}_{k_t^{(j)}}(b_i'^{(j)}),\,\mathrm{tag}_i^{(j)}\big) = 1.
		\end{equation}
		If verification fails, the corrupted message is immediately discarded and the protocol halts.
		Since $\textit{seed}_t$ evolves with the shuffle round index, no MAC key is ever reused, 
		and the verification process leaks no information about client identities or the global permutation $\pi$.
		
		\textbf{Protocol Integration.}
		The streaming MAC mechanism is embedded directly into the MPC shuffle pipeline.
		When $S_0$ and $S_1$ perform the re-masking operation 
		$b_i'^{(0)} = \pi_0(b_i^{(0)}) + r_0[i]$ and $b_i'^{(1)} = \pi_1(b_i^{(1)}) + r_1[i]$, 
		each re-masked share is immediately accompanied by its MAC tag.  
		Upon receiving a batch of shares, each server verifies the tags on-the-fly before proceeding to the next step of robust aggregation.
		This pipeline enables real-time detection of tampering or message loss within the same communication round, 
		eliminating any offline MAC check.
		The verified shares are then directly input to the Hamming–ReLU aggregation stage of Algorithm~\ref{alg:rain-ssa}.
		
		In Camel, post-shuffle MAC verification allows a malicious server to mount selective-failure attacks 
		by injecting small perturbations into shuffled entries and observing which checks fail.  
		RAIN prevents such leakage by verifying every re-masked share at streaming granularity.
		A corrupted share cannot reach the aggregation stage and thus reveals nothing about its position.  
		Formally, if one server attempts to bias a single share $b_i'^{(j)}$ by adding $\delta$, 
		the verification equation fails with probability $1 - 1/p$, 
		and the deviation is detected before any reconstruction.
		Therefore, the success probability of a selective-failure attack is negligible 
		and independent of the number of clients.
		
		All streaming MAC operations involve only field additions and multiplications over $\mathbb{Z}_p$, 
		introducing a constant-size tag per message.
		The additional communication and computation overhead remains below $1\%$ in all evaluated settings.  
		Because MAC generation and verification are performed on secret-shared values 
		using synchronized PRG seeds, the process is perfectly compatible with Shuffle-DP anonymity.
		As a result, RAIN achieves real-time integrity protection with negligible system cost 
		and without weakening privacy or scalability.
		
		\textbf{Malicious-Security Guarantee.}
		Combining two-server MPC computation with streaming MAC verification 
		achieves the same malicious-security level as prior verifiable shuffles but at linear cost.
		Any undetected deviation by one corrupted server would require either (i) forging a valid MAC tag 
		without knowledge of the complete key, or (ii) reconstructing inconsistent additive shares across servers.  
		Both events occur with negligible probability under standard MAC unforgeability and MPC correctness assumptions.
		Hence, RAIN’s shuffle-and-aggregation protocol is malicious-secure against any single colluding server:
		no message can be dropped, modified, or replayed without detection, 
		and every verified share is guaranteed to appear exactly once in the final aggregation.

		\section{Evaluation}
		\label{sec:Evaluation}
		
		\subsection{Experimental Setup}
		We evaluate our RAIN against both existing poisoning attacks and adaptive attacks in this section.
		
		\subsubsection{Datasets}
		
		We evaluate RAIN on multiple datasets from different domains, 
		including several image classification tasks and one human activity recognition task. 
		Following the standard practice in prior work~\cite{r24}, 
		we simulate heterogeneous data distributions among clients to reflect realistic FL scenarios. 
		Assume that each dataset contains $M$ classes; 
		we divide the clients into $M$ groups and allocate training examples probabilistically. 
		Each example with label $l$ is assigned to the group $l$ with probability $q>0$, 
		and to other groups with probability $\frac{1-q}{M-1}$. 
		Within each group, data are uniformly distributed to clients. 
		The parameter $q$ controls the non-IID degree of local data: 
		when $q = 1/M$, local data are IID; when $q > 1/M$, 
		the training data become increasingly skewed across clients. 
		Since non-IID data commonly occur in real FL deployments, 
		we adopt $q > 1/M$ by default to emulate such heterogeneity.
		
		\emph{MNIST.}
		The MNIST dataset~\cite{mnist} is a 10-class digit image classification dataset, which consists of 60,000 training examples and 10,000 testing examples. We set q = 0.1 in MNIST-0.1, which indicates local training data are IID among clients. We use MNIST-0.1 to show that RAIN is also effective in the IID setting.

		\emph{Fashion-MNIST.}
		The Fashion-MNIST dataset~\cite{fmnist} contains 60{,}000 training and 10{,}000 test images 
		across 10 clothing categories. 
		We distribute the samples to clients using $q = 0.5$ to simulate non-IID local data, 
		consistent with MNIST-0.5.
		
		\emph{CIFAR-10.}
		The CIFAR-10 dataset~\cite{cifar} includes 50{,}000 color images for training and 10{,}000 for testing, 
		spanning 10 object classes. 
		To emulate heterogeneous data distribution, 
		we again set $q = 0.5$ for client partitioning.

		\subsubsection{Poisoning Attacks for Evaluation}

		To comprehensively evaluate the robustness of RAIN,
		we consider both classical model poisoning attacks 
		and a recent adaptive attack specifically designed for DP-FL settings.
		
		\emph{Label Flipping (LF) Attack.}
		Following~\cite{r24}, each malicious client flips the label $l$ of its local samples 
		to $M - l - 1$, where $M$ is the total number of classes and $l \in \{0, 1, \ldots, M-1\}$.
		This untargeted attack misleads the global model by reversing label semantics.
		
		\emph{Krum Attack.}
		Krum~\cite{r15} is an untargeted local model poisoning attack 
		optimized for the Krum aggregation rule.
		Each malicious client uploads manipulated gradients crafted to be close to other malicious updates
		but far from benign ones, thereby biasing the aggregation result.
		
		\emph{Trim Attack.}
		Trim attack~\cite{r15} targets the Trim-mean and Median aggregation rules
		by injecting adversarial gradients that evade trimming thresholds.
		We adopt the same default parameters as in~\cite{r15}.
		
		\emph{Scaling Attack.}
		Following~\cite{r5}, the scaling attack is a targeted backdoor-style poisoning method.
		Each malicious client first augments a fraction $p$ of its local samples 
		with a predefined feature trigger and assigns them to an attacker-chosen label.
		The malicious client then scales its local model update by a large amplification factor $\lambda \gg 1$
		before uploading it to the server, causing the backdoor to dominate the aggregated gradient.
		
		\emph{Adaptive Attack on DP-FL(Attack-DPFL).}
		Beyond conventional poisoning settings, 
		we further evaluate an adaptive attack tailored to differentially private federated learning~\cite{r1}.
		In DP-FL, benign clients perturb their gradients with Gaussian noise, 
		resulting in large-norm perturbed gradients.
		Attack-DPFL exploits this property by re-scaling unperturbed poisoned gradients 
		to match the norm distribution of the noisy benign updates:
		\begin{equation}
			A_u = A \, G_u, \quad
			A = \frac{\|S_u\|}{\|G_u\|}, \quad
			S_u = h(G_u; l) + \mathcal{N}(0, \sigma^2).
		\end{equation}
		
		This alignment allows malicious updates to dominate the global aggregation 
		while bypassing the DP noise regularization.
		Compared with the conventional scaling or label-flipping attacks, 
		Attack-DPFL represents a stronger adaptive adversary 
		that specifically targets the robustness gap introduced by DP perturbation.
		RAIN is evaluated against all the above attacks 
		to demonstrate its resilience under both classical and DP-aware poisoning threats.

		\subsection{Utility}

		\begin{figure}[b]
			\centering
			% ---------- 第一行：单图居中 ----------
			\subfloat[FMNIST]{%
				\includegraphics[width=0.47\linewidth]{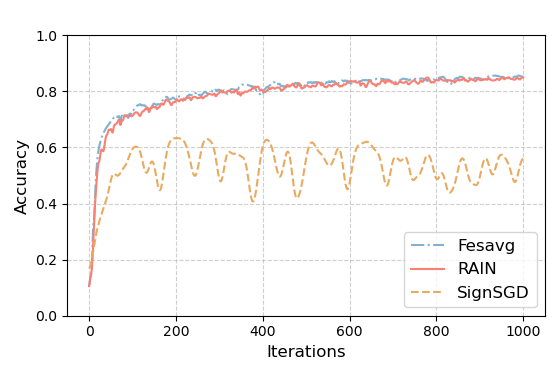}%
				\label{fig:fmnist_convergence}
			}\hfill
			% ---------- 第二行：两图并列 ----------
			\subfloat[MNIST]{%
				\includegraphics[width=0.47\linewidth]{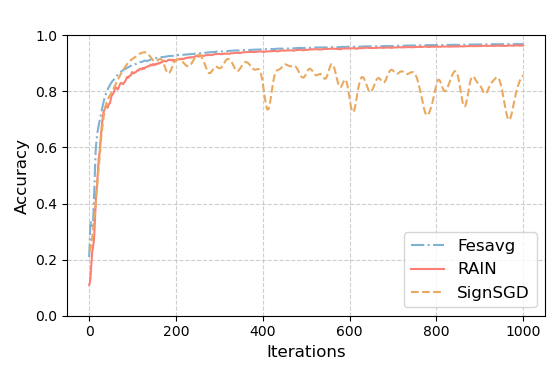}%
				\label{fig:mnist_convergence}
			}
			%	\subfloat[HAR]{%
				%		\includegraphics[width=0.49\linewidth]{har_convergence}%
				%		\label{fig:convergence-har}
				%	}
			
			\vspace{-1mm}
			\caption{
				\textbf{Convergence of RAIN under Shuffle-DP ($\rho = 0$).}
				Model accuracy (MA) over training iterations on FMNIST,and MNIST.
				RAIN achieves convergence comparable to FedAvg while significantly
				outperforming SignSGD in terms of stability and final accuracy across
				all datasets.
			}
			\label{fig:convergence}
			
		\end{figure}

		\subsubsection{Convergence under Shuffle-DP}
		The first experiment aims to verify that the proposed sign-space aggregation mechanism 
		does not hinder model convergence. 
		Since RAIN performs all robust aggregation operations directly in the 
		sign domain rather than in the real-valued gradient space,
		it is essential to confirm that the loss of magnitude information does not impair 
		training stability or convergence speed.
		We therefore evaluate RAIN in a benign setting ($\rho = 0$) 
		and without privacy noise,
		to isolate the impact of sign-space aggregation on model optimization.

		Figure~\ref{fig:fmnist_convergence} and Figure~\ref{fig:mnist_convergence} 
		show the evolution of model accuracy (MA) with training iterations 
		on FMNIST and MNIST datasets, respectively. 
		We compare RAIN with FedAvg (the canonical aggregation baseline) 
		and SignSGD \cite{r32}(a purely sign-based method without robustness adjustment).

		Across both datasets, RAIN converges to a similar MA level as FedAvg within the same number of training iterations on both FMNIST and MNIST datasets. This is because RAIN employs Hamming-distance–based trust scoring to adaptively weight sign-encoded updates, allowing the aggregator to approximate the behavior of full-precision averaging when $\delta = 0$. Although the sign encoding inevitably discards magnitude information, its impact on overall optimization is negligible, since the global direction of gradients is still preserved across clients. This empirical stability directly corroborates Lemma~\ref{lem:rain-tradeoff}, which predicts that the expected directional cosine between the aggregated sign vector and the true gradient remains above $0.7$ when $\varepsilon \!\ge\! 2$, ensuring convergence comparable to non-privatized training.

		In contrast, SignSGD converges much slower than both FedAvg and RAIN. 
		The reason is that SignSGD aggregates only the global sign of the summed gradients, 
		which behaves similarly to a coordinate-wise median operator in the sign domain. 
		Such a median-style rule implicitly excludes nearly $K-1$ values per coordinate in each round, 
		leading to severe information loss and oscillation in the global update direction. 
		Consequently, SignSGD introduces substantial degradation in convergence speed and final accuracy, 
		while RAIN maintains stable and efficient training dynamics comparable to FedAvg.
		
		\begin{figure}[h]
		\centering
		\subfloat[]{%
			\includegraphics[width=0.49\linewidth]{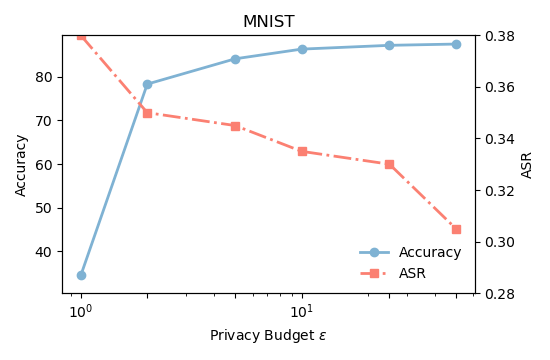}%
		}\hfill
		\subfloat[]{%
			\includegraphics[width=0.49\linewidth]{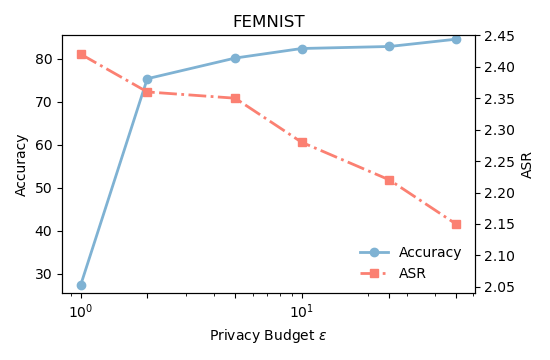}%
		}
		\vspace{-1mm}
		\caption{
			\textbf{Privacy–Robustness Trade-off of RAIN.}
			As the privacy budget $\varepsilon$ decreases, stronger DP noise slightly reduces model accuracy 
			but also lowers the attack success rate (ASR). 
		}
		\label{fig:rain-privacy-tradeoff}
		\vspace{-2mm}
	\end{figure}

		\subsubsection{Accuracy, Privacy, and Robustness Trade-off}
		
		\begin{figure*}[t]
			\centering
			\includegraphics[width=\textwidth]{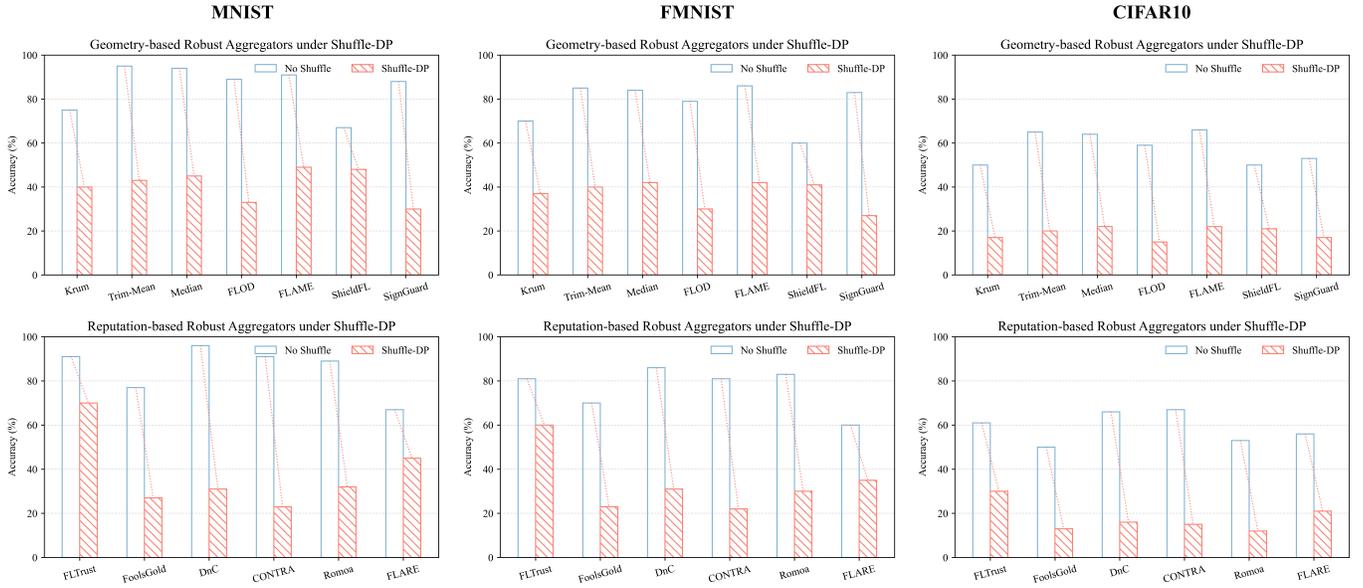}
			\caption{
				\textbf{Failure of existing robust aggregation methods under Shuffle-DP}
				Geometry-based methods (Krum, Trim-Mean, FLOD, FLAME, etc.) and 
				reputation-based methods (FLTrust, FoolsGold, DnC, CONTRA, FLARE, etc.) 
				are evaluated under the adaptive Attack-DPFL across three datasets.
			}
			\label{fig_4}
		\end{figure*}
		
		We further investigate the trade-off among accuracy, privacy, and robustness in RAIN
		by varying the privacy budget $\varepsilon$ under the Shuffle-DP setting. 
		As shown in Fig.~\ref{fig:rain-privacy-tradeoff}, 
		when the privacy guarantee becomes stronger (i.e., smaller $\varepsilon$), 
		model accuracy slightly decreases due to the heavier Gaussian noise injected during local randomization. 
		This degradation is expected since the stronger DP protection inevitably perturbs gradient signals. 
		However, RAIN still maintains competitive accuracy across all privacy levels, 
		showing that the sign-space aggregation is inherently robust to noise. 
		
		In addition, as $\varepsilon$ decreases, the attack success rate (ASR) also gradually declines, 
		indicating that stronger noise not only enhances privacy but can also suppress adversarial signal injection. 
		Nevertheless, excessive noise weakens the distinguishability between benign and malicious updates, 
		causing a mild loss in robustness when privacy is overly stringent. 
		Overall, these results demonstrate that RAIN achieves a favorable privacy–utility–robustness balance: 
		even under tight privacy budgets ($\varepsilon \le 5$), it retains over 70\% accuracy 
		while keeping ASR below 40\%, providing practical guidance for deploying Shuffle-DP FL systems.

		\subsection{Robustness}

		This section investigates the robustness of federated learning under adversarial conditions within the Shuffle-DP framework. 
		We first demonstrate that existing robust aggregation methods, both geometry-based (e.g., Krum, FLAME, FLOD) and reputation-based (e.g., FLTrust, FoolsGold, CONTRA), completely fail once client updates are anonymized and noised. 
		We then evaluate RAIN against state-of-the-art baselines under multiple poisoning and adaptive attacks across MNIST, FMNIST, and CIFAR-10.

		\subsubsection{Failure of Existing Robust Aggregation Methods.}
		We first investigate whether existing robust aggregation schemes remain effective 
		under the Shuffle-DP setting. 
		Figure~\ref{fig_4} presents the model accuracy under 
		the adaptive Attack-DPFL when 20\% of clients are malicious. 
		We evaluate both geometry-based methods (Krum, Trim-Mean, Median, FLOD, FLAME, ShieldFL, and SignGuard) 
		and reputation-based methods (FLTrust, FoolsGold, DnC, CONTRA, Romoa, and FLARE) 
		across MNIST, FMNIST, and CIFAR-10 datasets.
		
		The results show that nearly all existing robust FL approaches fail 
		once Shuffle-DP anonymization is applied.
		For geometry-based methods, local noise injection distorts the geometric relationships 
		among gradients, preventing distance- or angle-based filtering from identifying malicious updates. 
		Consequently, the test accuracy drops dramatically, often by more than 40–50\%.
		For example, FLAME and FLOD achieve around 90\% accuracy in the non-shuffled setting 
		but fall below 45\% under Shuffle-DP. 
		Similarly, reputation-based methods—whose effectiveness relies on persistent identity linkage 
		and temporal trust accumulation—lose the ability to track client behavior once the shuffler 
		removes client identifiers. 
		As a result, methods such as FLTrust and FoolsGold degrade to below 60\% accuracy.This experiment highlights a fundamental incompatibility between existing robust aggregation 
		mechanisms and the Shuffle-DP paradigm. 
		%When both geometric consistency and identity continuity are destroyed by shuffling and noise, 
		%the server can no longer infer which updates are trustworthy, 
		%causing traditional defenses to collapse under even moderate adversarial participation.
		
		\begin{figure}[t]
			\centering
			\includegraphics[width=1.05\linewidth]{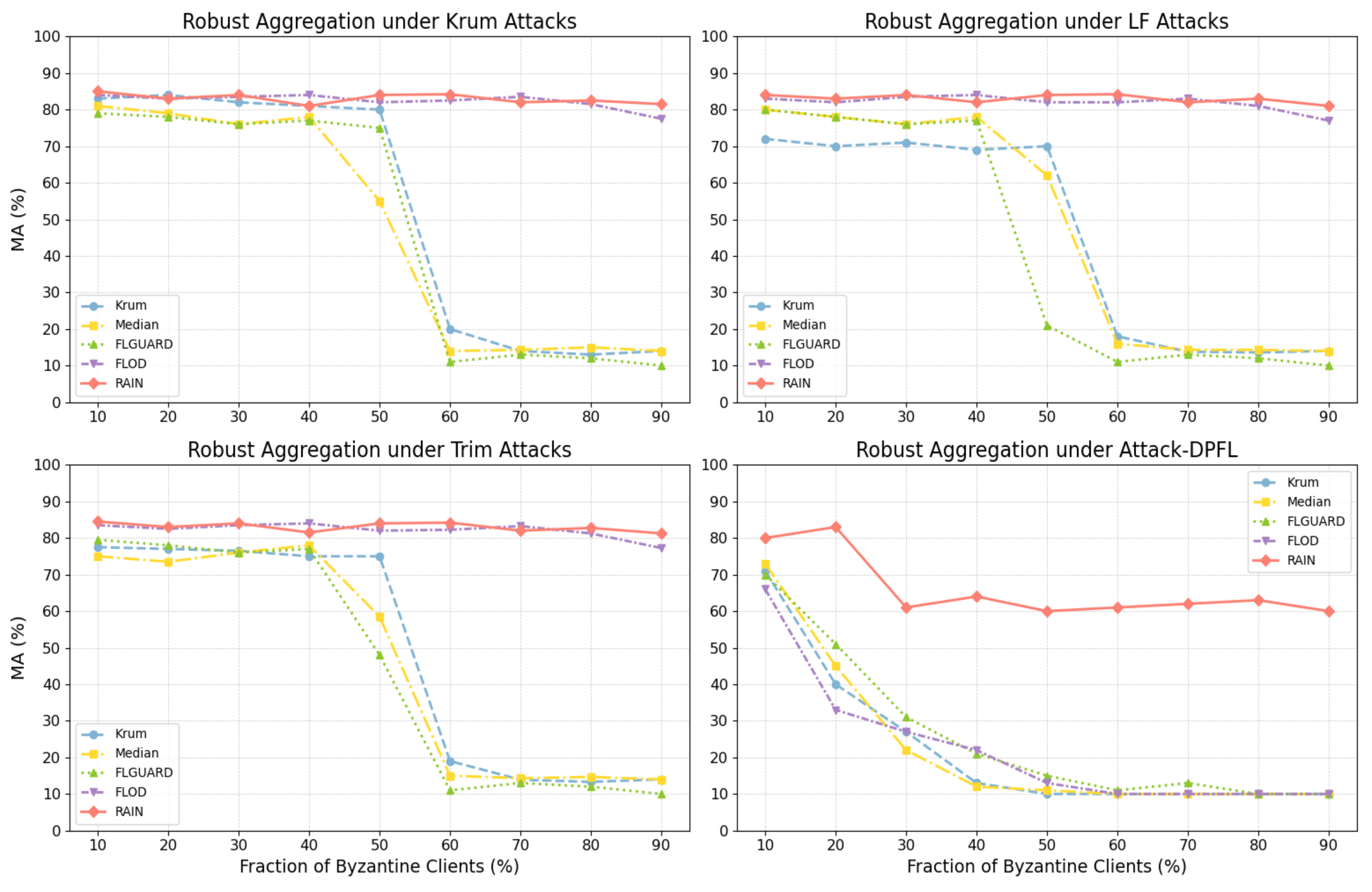}
			\caption{
				\textbf{Robust Aggregation of RAIN under Different Adversarial Attacks.}
				Each plot reports model accuracy as the fraction of Byzantine clients increases 
				from 10\% to 90\%. 
			}
			\label{fig:rain-mnist-attacks}
		\end{figure}

		\subsubsection{Impact of the number of malicious clients.}
		We further investigate the robustness of RAIN under various Byzantine and poisoning attacks on the MNIST dataset, 
		including Krum attack, Trim attack, Label Flipping (LF) attack, and the adaptive Attack-DPFL. 
		These attacks jointly cover both untargeted and targeted poisoning behaviors, 
		evaluating the model’s ability to resist malicious updates under different adversarial strategies. 
		As shown in Fig.~\ref{fig:rain-mnist-attacks}, 
		traditional robust aggregation schemes such as Krum, Median, FLGuard, and FLOD 
		exhibit severe performance degradation as the fraction of Byzantine clients increases. 
		Their model accuracy sharply drops below 20\% once the malicious ratio exceeds 50\%, 
		and continues to deteriorate as more adversaries join. 
		This confirms that existing geometry-based or reputation-based defenses 
		rely heavily on clean gradient structures or persistent client identities—both of which are destroyed by local perturbation and anonymization in Shuffle-DP. 
		
		In stark contrast, RAIN maintains consistently high accuracy across all attack types, 
		achieving around 80\% MA even when up to 90\% of clients behave maliciously. 
		This stability demonstrates that RAIN effectively mitigates both Byzantine manipulations 
		and adaptive poisoning by operating in the sign space rather than the raw gradient space. 
		Its Hamming-distance–based trust scoring bounds the influence of extreme or inconsistent sign updates, 
		allowing the aggregation process to remain stable even under heavy noise and anonymization. 
		Overall, these results highlight that RAIN achieves significantly stronger resilience 
		than state-of-the-art baselines while maintaining high utility in adversarial Shuffle-DP environments.

		\subsection{Efficiency}
		
		\subsubsection{Efficiency Analysis}

		We evaluate the communication efficiency of RAIN
		and compare it with representative secure Shuffle-DP frameworks,
		including FLGuard and Camel, on the MNIST and FMNIST datasets.
		Table~\ref{tab:comm_runtime_cost} summarizes the per-client communication cost
		and the server-side overall runtime,
		while Fig.~\ref{fig:rain-comm} further illustrates the communication breakdown
		under different communication paths and federation scales.

		\begin{table}[b]
			\centering
			\caption{Communication and runtime cost comparison}
			\label{tab:comm_runtime_cost}
			\begin{tabular}{lcccc}
				\toprule
				Method & Dataset &
				\makecell{Per-Client \\ Comm. Cost (KB)} &
				\makecell{Server-Side \\ Overall Runtime (s)} \\
				\midrule
				Camel~\cite{r33}   & MNIST  & 25.76 & 0.954 \\
				& FMNIST & 28.13 & 0.989 \\
				FLguard~\cite{r8} & MNIST  & 64.32 & 5.264 \\
				& FMNIST & 823.6 & 7.486 \\
				RAIN              & MNIST  & 25.30 & 0.910 \\
				& FMNIST & 76.43 & 0.980 \\
				\bottomrule
			\end{tabular}
		\end{table}
		
		For client-to-server (C--P) communication,
		Table~\ref{tab:comm_runtime_cost} shows that RAIN
		significantly reduces the per-client bandwidth consumption compared to FLGuard.
		On MNIST, RAIN requires 25.30\,KB per client per round,
		achieving a 60.7\% reduction relative to FLGuard (64.32\,KB),
		and exhibits communication cost comparable to Camel (25.76\,KB).
		On FMNIST, RAIN reduces the per-client communication volume by 90.7\%
		compared to FLGuard (823.6\,KB), while incurring higher overhead than Camel,
		due to differences in model dimensionality and encoding granularity.
		Overall, these results indicate that sign-space representation
		effectively lowers client-side communication compared to
		floating-point–based secure aggregation schemes.
		
		\begin{figure}[t]
			\centering
			% ---------- 第一行：单图居中 ----------
			\subfloat[C–P Communication]{%
				\includegraphics[width=0.50\linewidth]{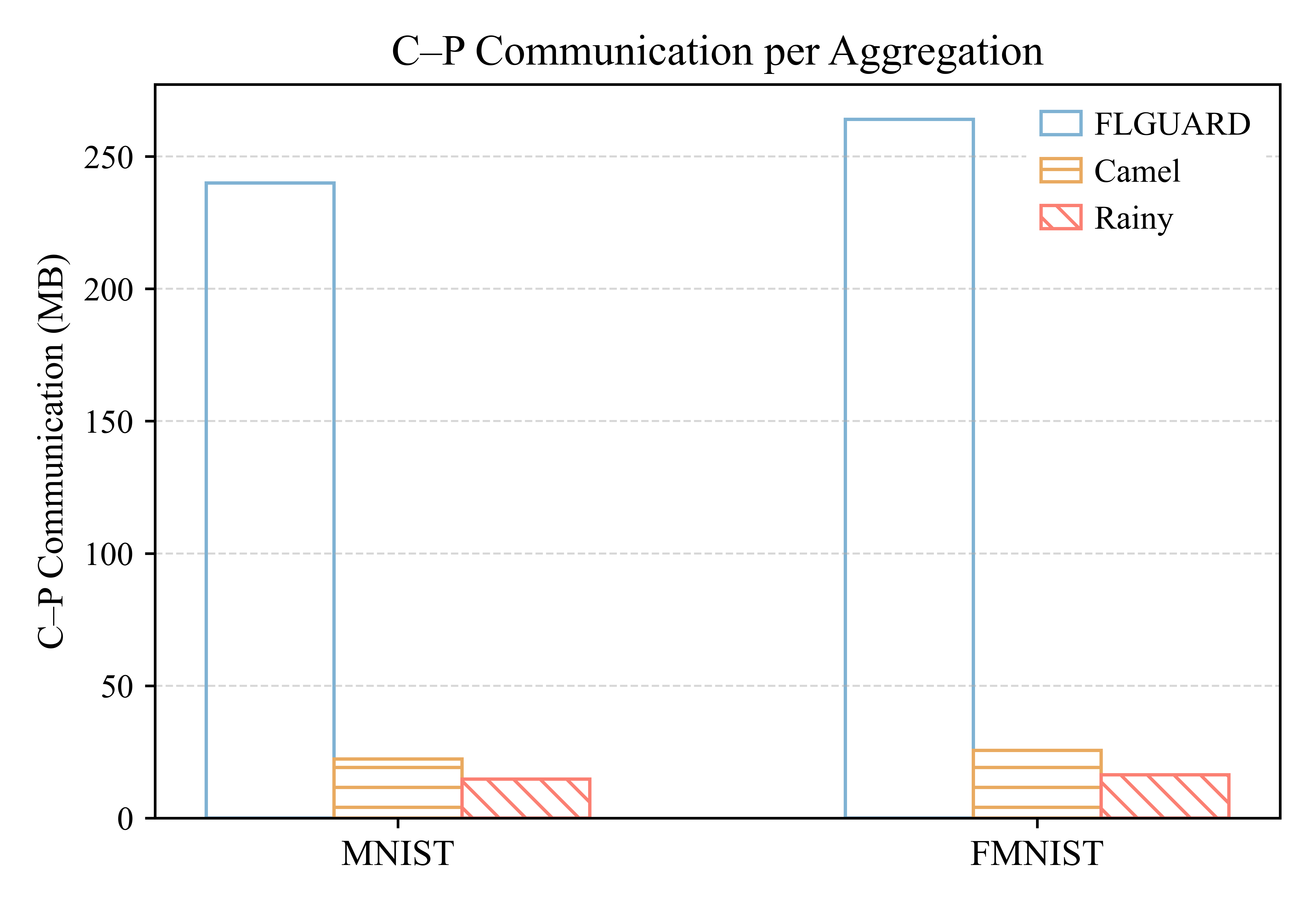}%
			}\\[-1mm]
			% ---------- 第二行：两图并列 ----------
			\subfloat[P–P Communication (MNIST)]{%
				\includegraphics[width=0.49\linewidth]{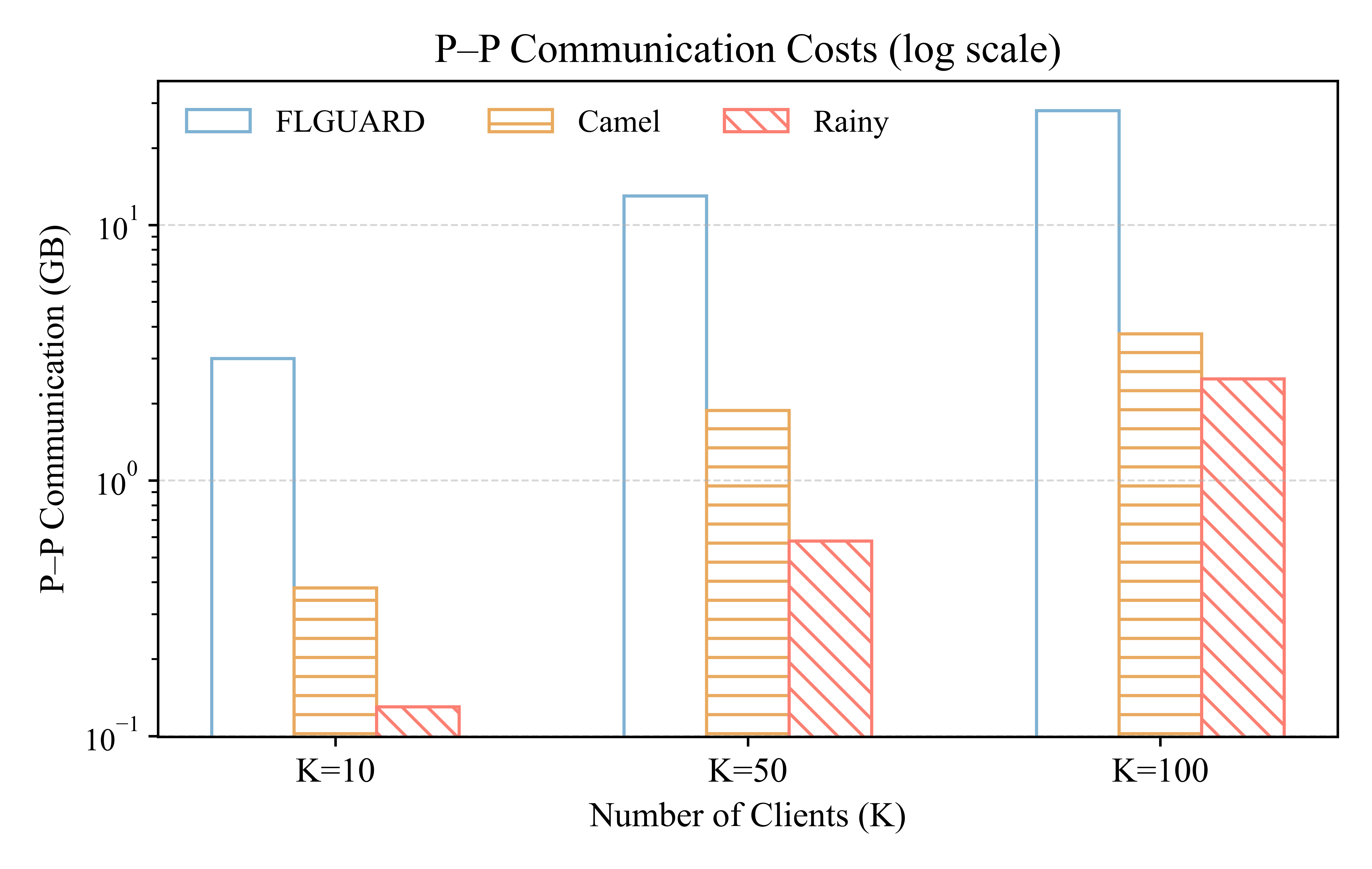}%
			}\hfill
			\subfloat[P–P Communication (FMNIST)]{%
				\includegraphics[width=0.49\linewidth]{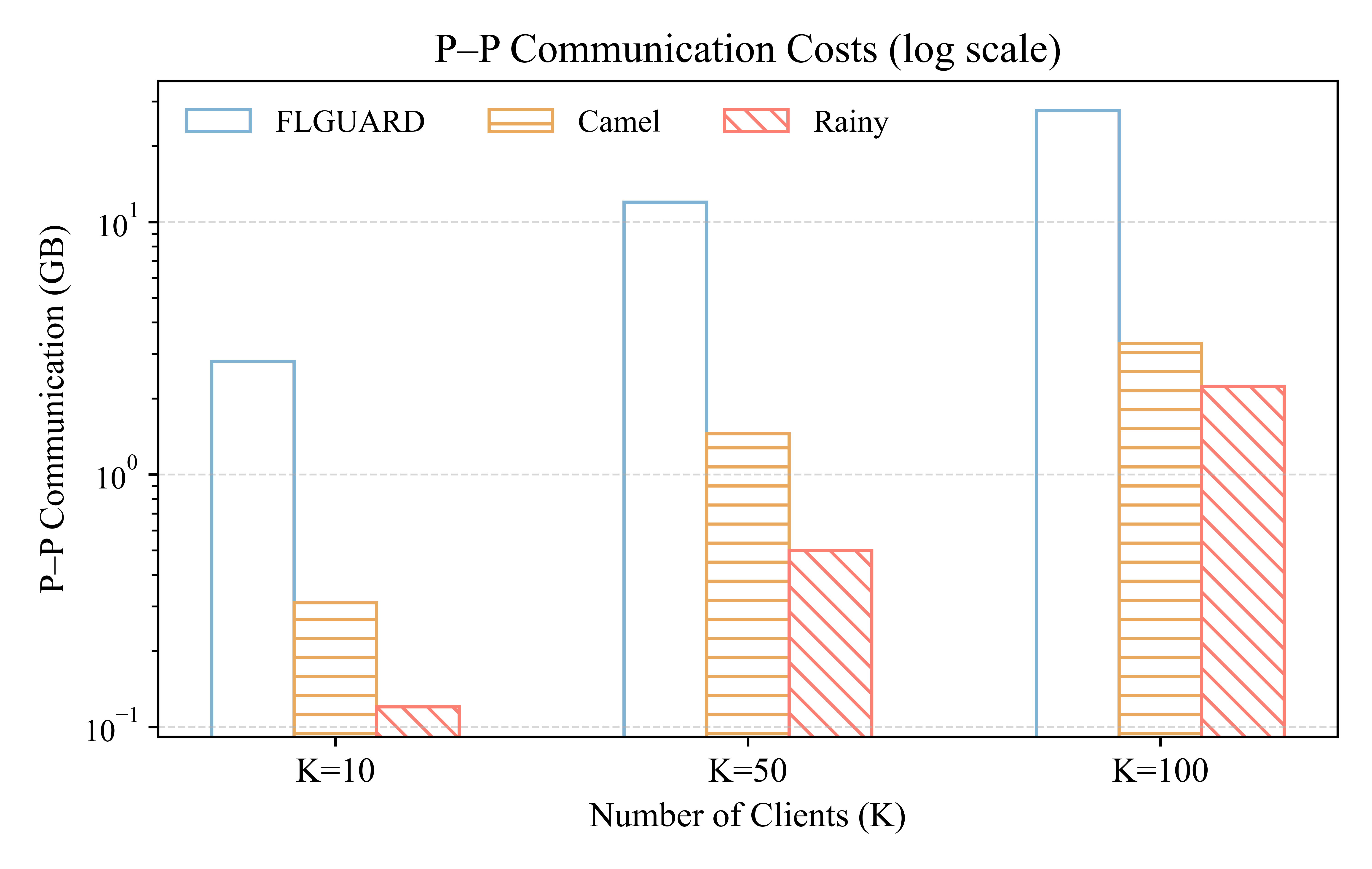}%
			}
			
			\vspace{-1mm}
			\caption{
				\textbf{Communication Efficiency of RAIN.}
				(a) Client-to-server (C–P) communication per aggregation on MNIST and FMNIST. 
				(b–c) Server-to-server (P–P) communication under different numbers of clients $K$ (log scale). 
			}
			\label{fig:rain-comm}
			\vspace{-2mm}
			
		\end{figure}
		
		For server-to-server (P--P) communication,
		Fig.~\ref{fig:rain-comm}(b--c) show that RAIN exhibits sub-linear growth
		with respect to the number of clients $K$,
		remaining in the megabyte (MB) range even for large-scale federations.
		In contrast, both Camel and FLGuard incur gigabyte-level (GB) communication
		overhead as $K$ increases,
		due to repeated re-masking, clustering, and multi-round coordination
		among aggregation servers.
		This demonstrates that RAIN offers substantially improved communication
		scalability under Shuffle-DP.

		We further evaluate the computation efficiency of RAIN
		and compare it with representative secure federated learning frameworks,
		including FLGuard and Camel.
		Table~\ref{tab:comp_cost} reports the per-client and server-side computation cost
		on the MNIST and FMNIST datasets.
		
		\begin{table}[t]
			\centering
			\caption{Computation cost comparison}
			\label{tab:comp_cost}
			\begin{tabular}{lcccc}
				\toprule
				Method & Dataset &
				\makecell{Per-Client \\ Comp. Cost (s)} &
				\makecell{Server-Side \\ Comp. Cost (s)} \\
				\midrule
				Camel~\cite{r33}   & MNIST  & 0.046 & 0.031 \\
				& FMNIST & 0.059 & 0.154 \\
				FLguard~\cite{r8} & MNIST  & 0.210 & 0.420 \\
				& FMNIST & 0.636 & 0.470 \\
				RAIN              & MNIST  & 0.055 & 0.052 \\
				& FMNIST & 0.063 & 0.067 \\
				\bottomrule
			\end{tabular}
		\end{table}

		Specifically, the per-client computation cost of RAIN 
		is only 0.055\,s on MNIST and 0.063\,s on FMNIST—approximately $4$–$5\times$ lower than FLGuard. 
		This efficiency gain is mainly attributed to the streamlined sign-space aggregation, 
		which eliminates expensive cryptographic operations during local updates. 
		On the server side, RAIN's computation cost remains lightweight ($\approx$0.05–0.07\,s) 
		due to its two-server MPC design, which performs secure shuffling and aggregation directly on additive shares 
		without costly zero-knowledge or re-masking procedures. 
		As a result, RAIN reduces the overall runtime by more than $7\times$ compared to FLGuard, 
		demonstrating its practicality for large-scale federated deployments under Shuffle-DP.

		\section{Discussion}
		
		Our analysis shows that RAIN achieves stable aggregation under
		moderate privacy budgets, where noise-induced sign flips remain
		bounded, indicating the existence of a practical operating
		regime for robust aggregation under Shuffle-DP.
		A natural direction for future work is to more precisely
		characterize this regime by establishing tighter relationships
		between privacy budgets, local noise, and DP-aware adversarial
		strategies.
		While such an analysis is feasible within the RAIN framework,
		it is left for future work due to scope and time constraints.
		
		To focus on this core interaction between privacy and robustness,
		the current study assumes independent and identically
		distributed (IID) client data.
		Extending RAIN to non-IID settings and to more realistic
		deployment scenarios, including asynchronous participation and
		stronger server-side adversaries, remains an important and
		promising direction for future research.
		
		\section{Conclusion}
		
		This work demonstrates that the shuffle model of differential privacy
		fundamentally breaks the assumptions underpinning existing geometry and
		reputation based robust aggregation defenses.
		Local noise erases meaningful geometric relations among updates, while shuffling
		destroys identity continuity, causing prior approaches to collapse under
		adaptive and Byzantine attacks.
		Taken together, these findings reveal a critical gap between privacy
		amplification and adversarial robustness in widely adopted Shuffle-DP pipelines.
		
		To bridge this gap, we propose RAIN, a unified framework that reconciles privacy,
		robustness, and verifiable integrity under Shuffle-DP.
		RAIN departs from geometry-dependent aggregation and instead operates in the sign
		space, leveraging consistency-based trust scoring to bound malicious influence
		without relying on client identities.
		By instantiating this mechanism within a secret-shared shuffle-and-aggregation pipeline and enforcing streaming integrity verification, RAIN achieves malicious
		security while preserving Shuffle-DP’s privacy guarantees and avoiding the heavy overheads typically associated with public-key or zero-knowledge techniques.
		
		Extensive evaluations across standard benchmarks show that RAIN is both secure
		and practical.
		It consistently withstands strong poisoning and adaptive attacks, maintains
		competitive model accuracy, and substantially improves communication efficiency
		and aggregation latency compared to prior robust defenses under Shuffle-DP.
		These results confirm that robust and verifiable federated learning can be
		achieved without sacrificing differential privacy.

		{\footnotesize
%			\nocite{*}  
			\bibliographystyle{IEEEtran}
			\bibliography{rainy}          % 不要写 ref.bib
		}
		
		\appendix
		
		\section{Sign-based Aggregation and Hamming Similarity Analysis}
		\label{app:sign_hamming}
		
		\subsection{Convergence of Sign-based Aggregation}
		
		Let $g_i \in \mathbb{R}^d$ be the local gradient of client $i$ and $g^*$ be the true global gradient.
		Each client uploads a sign-encoded vector $\tilde{g}_i = \operatorname{sign}(g_i + \mathcal{N}(0,\sigma^2))$
		under local randomization.
		Prior works~\cite{r37,r40} show that when each coordinate of $g_i$ independently satisfies
		$\Pr[\operatorname{sign}(\tilde{g}_{ij}) = \operatorname{sign}(g_j^*)] > \tfrac{1}{2}$,
		the aggregated sign $\tilde{g}_{\text{agg}} = \operatorname{sign}(\sum_i \tilde{g}_i)$
		is an unbiased estimator of $\operatorname{sign}(g^*)$ in expectation:
		\[
		\mathbb{E}[\tilde{g}_{\text{agg}}^j] = (2p_j - 1) \operatorname{sign}(g_j^*), \quad
		p_j = \Pr[\operatorname{sign}(\tilde{g}_{ij}) = \operatorname{sign}(g_j^*)].
		\]
		Thus, as long as the majority of clients are benign ($p_j > 0.5$), 
		the expected update direction remains aligned with the true gradient, 
		and convergence toward the global optimum is guaranteed in the same asymptotic rate as
		standard stochastic gradient descent~\cite{r32}.
		Formally,
		\[
		\mathbb{E}[f(w_{t+1}) - f(w^*)] \leq (1 - \eta\mu)^t \big(f(w_0) - f(w^*)\big) + O(\eta\sigma^2),
		\]
		where $\eta$ is the learning rate and $\mu$ is the smoothness parameter.
		This indicates that sign encoding, even under random noise and shuffling, 
		preserves convergence up to a bounded constant term.

		\subsection{Hamming Distance as a Robust Similarity Metric}
		
		To quantify the similarity between two sign vectors 
		$\tilde{w}_i, \tilde{w}_s \in \{-1,1\}^d$,
		we start from the cosine similarity:
		\[
		c_i = \frac{\tilde{w}_i \cdot \tilde{w}_s}{\|\tilde{w}_i\|\|\tilde{w}_s\|} 
		= \frac{1}{d}\sum_{j=1}^d \tilde{w}_{ij}\tilde{w}_{sj}.
		\]
		Since each coordinate takes values in $\{-1,1\}$, we can rewrite:
		\begin{align*}
			c_i 
			&= \frac{1}{d}\sum_{j=1}^{d}(1-2\mathbb{E}(\tilde{w}_{ij}))\cdot(1-2\mathbb{E}(\tilde{w}_{sj})) \\
			&= 1 - \frac{2}{d}\sum_{j=1}^{d}\big(\mathbb{E}(\tilde{w}_{ij}) + \mathbb{E}(\tilde{w}_{sj}) 
			- 2\mathbb{E}(\tilde{w}_{ij})\mathbb{E}(\tilde{w}_{sj})\big) \\
			&= 1 - \frac{2}{d}\sum_{j=1}^{d}\mathbb{E}(\tilde{w}_{ij}) \oplus \mathbb{E}(\tilde{w}_{sj}) \\
			&= 1 - 2\frac{\text{hd}_i}{d},
		\end{align*}
		where $\text{hd}_i = \sum_{j=1}^{d}\mathbf{1}[\tilde{w}_{ij}\neq \tilde{w}_{sj}]$
		is the Hamming distance between the two sign vectors.
		
		Thus, $c_i > 0 \Leftrightarrow 1 - 2\text{hd}_i/d > 0 \Leftrightarrow \text{hd}_i < d/2$.
		Therefore, with threshold $\tau = d/2$, we have
		\[
		\nu_i > 0 \Leftrightarrow c_i > 0,
		\]
		which means that a $\tau$-clipping Hamming distance–based rule is equivalent 
		to cosine-similarity clipping in detecting poisoned sign updates.
		Moreover, our $\tau$-clipping formulation
		\[
		w_i = \operatorname{ReLU}(\tau - \text{hd}_i)
		\]
		generalizes cosine-based clipping (e.g., FLTrust) by allowing tunable $\tau$,
		making the Hamming-distance metric more flexible and adaptive across datasets and noise levels.
	\end{document}